\documentclass[10pt,letterpaper]{article}

\usepackage{mathrsfs}
\usepackage{fancyhdr}
\usepackage{latexsym}
\usepackage[dvips]{graphics}
\usepackage{graphics}
\usepackage{graphicx}
\usepackage{epsfig}
\usepackage{bm}
\usepackage{amsmath,amsfonts,amsthm,amssymb,amscd}
\input amssym.def
\input amssym.tex
\usepackage{color}
\usepackage{hyperref}
\usepackage{ulem}

\newcommand\be{\begin{equation}}
\newcommand\en{\end{equation}}
\newcommand\bes{\begin{equation*}}
\newcommand\ees{\end{equation*}}

\newcommand\bea{\begin{eqnarray}}
\newcommand\eea{\end{eqnarray}}
\newcommand\bi{\begin{itemize}}
\newcommand\ei{\end{itemize}}
\newcommand\ben{\begin{enumerate}}
\newcommand\een{\end{enumerate}}

\theoremstyle{definition}
\newtheorem{thm}{Theorem}

\newtheorem{lem}{Lemma}
\newtheorem{cor}[lem]{Corollary}
\newtheorem{prop}[lem]{Proposition}

\newtheorem{defi}[lem]{Definition}

\newtheorem{rek}[lem]{Remark}

\newtheorem*{clm}{Claim}

\newcommand{\R}{\ensuremath{\mathbb{R}}}

\newcommand{\Z}{\ensuremath{\mathbb{Z}}}

\newcommand{\N}{\mathbb{N}}

\newcommand{\matthree}[9]
{\left(\begin{array}{ccc}
                        #1  & #2 & #3  \\
                        #4 &  #5 & #6 \\
                        #7 &  #8 & #9
                          \end{array}\right) }

\newcommand{\la}{\lambda}
\newcommand{\La}{\Lambda}

\newcommand{\Lst}{\; \Big| \;}          
\newcommand{\ve}{\mathbf}
\newcommand{\p}{\partial}
\newcommand{\phiz}{\phi^0}
\newcommand{\psiz}{\psi^0}
\newcommand{\Proj}{\mathcal{P}}
\newcommand{\Aonel}{\mathcal{A}_1^\lambda}
\newcommand{\Aonelz}{\mathcal{A}_1^{\lambda_0}}
\newcommand{\Aonelnm}{\mathcal{A}_1^{\lambda_{n}}}
\newcommand{\Aonelk}{\mathcal{A}_1^{\lanmk}}

\newcommand{\Atwol}{\mathcal{A}_2^\lambda}

\newcommand{\Atwolnm}{\mathcal{A}_2^{\lambda_{n}}}
\newcommand{\Atwolk}{\mathcal{A}_2^{\lanmk}}

\newcommand{\Atwolz}{\mathcal{A}_2^{\lambda_{0}}}

\newcommand{\Bl}{\mathcal{B}^\lambda}
\newcommand{\Blz}{\mathcal{B}^{\lambda_0}}
\newcommand{\Bs}{\mathcal{B}^\sigma}
\newcommand{\Blnm}{\mathcal{B}^{\lambda_{n}}}
\newcommand{\Blk}{\mathcal{B}^{\lanmk}}

\newcommand{\Cl}{\mathcal{C}^\lambda}
\newcommand{\Clnm}{\mathcal{C}^{\lambda_{n}}}
\newcommand{\Clk}{\mathcal{C}^{\lanmk}}
\newcommand{\Clz}{\mathcal{C}^{\lambda_{0}}}
\newcommand{\Dl}{\mathcal{D}^\lambda}

\newcommand{\Dlnm}{\mathcal{D}^{\lambda_{n}}}
\newcommand{\Dlk}{\mathcal{D}^{\lanmk}}
\newcommand{\Dlz}{\mathcal{D}^{\lambda_{0}}}
\newcommand{\Aonez}{\mathcal{A}_1^0}
\newcommand{\Ail}{\mathcal{A}_i^{\lambda}}
\newcommand{\Ais}{\mathcal{A}_i^{\sigma}}
\newcommand{\Aoneze}{\mathcal{A}_1^{0,\epsilon}}

\newcommand{\Aones}{\mathcal{A}_1^\sigma}

\newcommand{\Atwoz}{\mathcal{A}_2^0}
\newcommand{\Atwoze}{\mathcal{A}_2^{0,\epsilon}}
\newcommand{\Atwozz}{\mathcal{A}_2^{0,0}}

\newcommand{\Lz}{\mathcal{L}^0}

\newcommand{\Ql}{\mathcal{Q}^\lambda}
\newcommand{\Qlnm}{\mathcal{Q}^{\lambda_{n}}}

\newcommand{\Qs}{\mathcal{Q}^\sigma}

\newcommand{\Mnl}{\mathcal{M}^\lambda_{n}}
\newcommand{\Mnz}{\mathcal{M}^0_{n}}

\newcommand{\Mns}{\mathcal{M}^\sigma_{n}}

\newcommand{\Mnmlnm}{\mathcal{M}^{\lambda_{n}}_{n}}
\newcommand{\Mklk}{\mathcal{M}^{\lanmk}_{n}}
\newcommand{\Ml}{\mathcal{M}^\lambda}
\newcommand{\Ms}{\mathcal{M}^\sigma}
\newcommand{\Mz}{\mathcal{M}^0}
\newcommand{\Mlz}{\mathcal{M}^{\lambda_0}}
\newcommand{\domain}{[0,P]\times\R^2}
\newcommand{\fnm}{\phi_{n}}
\newcommand{\fk}{\phi_{n}}

\newcommand{\pnm}{\psi_{n}}

\newcommand{\bnm}{b_{n}}
\newcommand{\hpt}{H_P^2}
\newcommand{\hptz}{H_{P,0}^2}
\newcommand{\lpt}{L_P^2}
\newcommand{\lptz}{L_{P,0}^2}
\newcommand{\lwt}{L_w^2}

\newcommand{\psiez}{\psi_\epsilon^0}
\newcommand{\Tpsiez}{T_{\psi_\epsilon^0}}

\newcommand{\lanmk}{\la_{n}}
\newcommand{\ula}{\la^*}
\newcommand{\sigs}{\sigma^*}

\numberwithin{lem}{section}
\numberwithin{equation}{section}

\begin{document}

\title{Instability of Nonmonotone Magnetic Equilibria of the Relativistic Vlasov-Maxwell System}
\author{Jonathan Ben-Artzi}
\maketitle

\begin{abstract}
We consider the question of linear instability of an equilibrium of the Relativistic Vlasov-Maxwell (RVM) System that has a strong magnetic field. Standard instability results deal with systems where there are fewer particles with higher energies. In this paper we extend those results to the class of equilibria for which the number of particles does not depend monotonically on the energy. Without the standard sign assumptions, the analysis becomes significantly more involved.
\end{abstract}

\section{Introduction}

\subsection{The Relativistic Vlasov-Maxwell System}
In this paper we consider the linear stability of a super-heated plasma. The plasma is assumed to have low density, and thus collisions between particles may be ignored. The behavior of such a system is governed by the Relativistic Vlasov-Maxwell (RVM) System of Equations
	\begin{subequations}\label{mainequations}
	\begin{equation}\label{mainvlasov}
	\p_tf^{\pm}+\hat{v}\cdot\nabla_xf^{\pm}\pm(\ve{E}+\ve{E}^{ext}+\hat{v}\times(\ve{B}+\ve{B}^{ext}))\cdot\nabla_vf^{\pm}=0,
	\end{equation}
	\begin{equation}
	\p_t\ve{E}=\nabla\times\ve{B}-\ve{j},
	\hspace{10pt}
	\nabla\cdot\ve{E}=\rho,
	\hspace{10pt}
	\rho=\int(f^+-f^-)\;dv,
	\end{equation}
	\begin{equation}
	\p_t\ve{B}=-\nabla\times\ve{E},
	\hspace{10pt}
	\nabla\cdot\ve{B}=0,
	\hspace{10pt}
	\ve{j}=\int\hat{v}\;(f^+-f^-)\;dv,
	\end{equation}
	\end{subequations}
where $v$ is the momentum, $\hat{v}=v/\sqrt{1+|v|^2}$ is the velocity and all physical constants (like the speed of light $c$ and the mass of the particles $m$) are set to be $1$. The superscripts $\pm$ indicate the two species -- electrons and ions. The transport equation \eqref{mainvlasov} is called the Vlasov Equation, and it is coupled with Maxwell's Equations. $\ve{E}(t,x)$ and $\ve{B}(t,x)$ are the electric and magnetic fields, $f^{\pm}(t,x,v)\geq0$ are the electron and ion distribution functions, and $\ve{E}^{ext}$ and $\ve{B}^{ext}$ are external electric and magnetic fields. In addition, $\rho(x)$ is the net charge density at the point $x$, and $\ve{j}(x)$ is the current density at the point $x$. Particle charges are always $+1$ or $-1$. The assumption that all masses are equal is physically relevant: It often occurs that interstellar plasma is made up of dust particles of uniform mass charged positively and negatively.

The subject of this paper is the ``purely magnetic'' case, where the electric field is assumed to be negligible compared to the magnetic field, at least on the typical scale of the system. Together with a certain symmetry assumption on the two species, this assumption makes certain aspects of the analysis simpler, while keeping most of the interesting physical aspects. This result is outlined in Section 9 of \cite{rvm2}, though in \cite{rvm2} the authors fail to consider all possible cases, and, therefore, their proof outline is incomplete.

Our equilibrium is represented in terms of two invariants of the particles -- the energy $e$ and the momentum $p$. We explore the question of linear instability, and give concrete criteria for instability in terms of the spectral properties of certain Schr\"{o}dinger operators acting on periodic functions of $x$ alone, not $v$.

\subsection{The $1\frac{1}{2}$ Dimensional Case}
For simplicity, we take the lowest dimensional system which has a nontrivial magnetic field and which has many physical applications: The so-called $1\frac{1}{2}$ dimensional case. In this setting, we have one spatial dimension and a two-dimensional momentum space. The single spatial variable $x$ corresponds to $v_1$, and the additional velocity dimension is denoted by $v_2$. We write $v=(v_1,v_2)$. It turns out that this dimensional setting is actually quite significant even for physicists: When one considers a tokamak (which typically has cylindrical symmetry), then locally, in a small region, the symmetries that one observes are precisely those described here.

The notation in the $1\frac{1}{2}$ dimensional case is as follows: We let $f^\pm(t,x,v)$ be the ion and electron distribution functions, $\ve{E}(t,x)=\left(E_1(t,x),E_2(t,x),0\right)$ and $\ve{B}(t,x)=(0,0,B(t,x))$ be the electric and magnetic fields. In addition, we define the electric and magnetic potentials $\phi$ and $\psi$, which satisfy
	\begin{equation}\label{potentials}
	\p_x\phi
	=
	-E_1
	\hspace{2 cm}
	\p_x\psi
	=
	B.
	\end{equation}

We assume the existence of an equilibrium $f^{0,\pm}(x,v)$ which is a solution of RVM. By Jeans' Theorem (cf \cite{dendy}) this equilibrium can be represented in the coordinates
	$$
	e^\pm
	=
	\left<v\right>\pm\phi^0(x)\hspace{1cm}
	p^\pm
	=
	v_2\pm\psi^0(x)
	$$
as $f^{0,\pm}(x,v)=\mu^\pm(e^\pm,p^\pm)$, where $\phi^0$ and $\psi^0$ are the electric and magnetic potentials, satisfying $E_1^0=-\p_x\phi^0$ and $B^0=\p_x\psi^0$, with $E^0_1$ and $B^0$ being the equilibrium electric and magnetic fields. In addition $E_2^0\equiv0$. We note that the purely magnetic assumption implies that $\phi^0\equiv0$, and, hence, $e^+=e^-$ can simply be denoted by $e$. Henceforth it will be understood that $\mu^\pm$ are evaluated at $e, p^\pm$, so we will simply write $\mu^\pm(e,p)$. In this paper, we consider the ``\emph{nonmonotone}'' case; by that we simply mean that
	\begin{equation}
	\mu^\pm_e:=\frac{\p\mu^\pm}{\p e}\nless0
	\end{equation}
on some subset of the set $\{\mu^\pm(e,p)>0\}$. Here, ``$\mu_e^\pm$" means ``the derivative of $\mu^\pm$ with respect to the first component evaluated at $(e,p^\pm)$." Similarly, ``$\mu_p^\pm$" is the derivative with respect to the second component. By ``monotone", we mean that $\mu^\pm_e<0$ on the set $\{\mu^\pm>0\}$. Roughly speaking, the coordinates $e$ and $p$ should be understood to be energy and momentum respectively.

\subsection{Main Results}\label{mainres}
We assume for simplicity that the equilibrium has some given period $P$. We define two operators acting on functions of $x$, whose properties will be rigorously treated later, as
	\begin{eqnarray*}
	\Aonez h&=&-\p_x^2h-\left(\sum_\pm\int\mu^\pm_e\;dv\right)h+\sum_\pm\int\mu^\pm_e\Proj^\pm h\;dv\\
	\Atwoz h&=&-\p_x^2h-\left(\sum_\pm\int\hat{v}_2\mu^\pm_p\;dv\right)h-\sum_\pm\int\mu^\pm_e\hat{v}_2\Proj^\pm (\hat{v}_2h)\;dv,
	\end{eqnarray*}
where $\Proj^\pm$ are projection operators onto some subspaces of a certain Hilbert space, also rigorously defined later. The operators $\Aonez$ and $\Atwoz$ have as their domain the space
	$$
	\hpt
	=
	\left\{h\text{ is }P\text{ periodic on }\R\text{ and }h\in H^2(0,P)\right\}
	$$
where $H^2(0,P)$ is the usual Sobolev space of functions on $(0,P)$ whose first two derivatives are square integrable. We also define the number $l^0$ as:
	$$
	l^0=\frac{1}{P}\sum_\pm\int_0^P\int\hat{v}_1\mu^\pm_e\Proj^\pm\left(\hat{v}_1\right)\;dv\;dx.
	$$
In order to properly define function spaces that include functions that do not necessarily decay at infinity, we define a weight $w$ that has the form
	\begin{equation}\label{weight}
	w
	=
	c(1+|e|)^{-\alpha}
	\end{equation}
for some $\alpha>2$ and $c>0$, where $e$ is the particle energy as defined above. We require $\mu^\pm\in C^1$ to satisfy
	\begin{equation}
	(|\mu^\pm_e|+|\mu^\pm_p|)(e,p)
	\leq
	w(e)\label{weight}
	\end{equation}
so that $\int(|\mu^\pm_e|+|\mu^\pm_p|)\;dv$ is finite. Obviously, we require $\mu^\pm(e,p)\geq0$.

As previously mentioned, we assume that the equilibrium is \emph{purely magnetic} -- that is, we assume that it is arranged in such a way so that there is no electric field, only a magnetic field -- and satisfies a symmetry property that ensures that $\rho^0=0$:
	\begin{equation}\label{pur-mag}
	\phi^0\equiv0
	\hspace{.5cm}\text{and}\hspace{.5cm}
	\mu^+(e,p)=\mu^-(e,-p)
	\end{equation}
Lin and Strauss prove the existence of such equilibria in the appendix of \cite{rvm1}. The assumption $\phi^0\equiv0$ is physical: In many physical systems the equilibrium state has negligible electric fields which are only significant on small scales, while the magnetic field is significant on the typical scale of the system.

In this paper, we denote
	$$
	neg(\mathcal{F})
	=
	\left\{
	\text{the number of negative eigenvalues (counting multiplicity) of the operator }\mathcal{F}
	\right\}.
	$$
Similarly, $pos(\mathcal{F})$ denotes the number of positive eigenvalues. Later we will show that $\Aonez$ and $\Atwoz$ have discrete spectra. For a real number $a$, $neg(a)$ is $1$ if $a<0$ and $0$ otherwise.

\begin{thm}\label{mainthm}
Let $f^{0,\pm}(x,v)=\mu^\pm(e,p)$ be a periodic equilibrium satisfying \eqref{weight} and \eqref{pur-mag}. Assume that the null space of $\Aonez$ consists of the constant functions and that $l^0\neq 0$. Then the equilibrium is spectrally unstable if
	$$neg\left(\Atwoz\right)>neg\left(\Aonez\right)+neg(-l^0).$$
\end{thm}

\begin{thm}\label{mainthm2}
Under the additional assumption that the null space of $\Atwoz$ is trivial, the equilibrium is spectrally unstable if
	$$neg\left(\Atwoz\right)\neq neg\left(\Aonez\right)+neg(-l^0).$$
\end{thm}

By ``spectrally unstable" we mean that the system linearized around the equilibrium solution $\mu^\pm(e,p)$ has a purely growing mode solution of the form
	\begin{equation}\label{purely}
	\left(e^{\la t}f^\pm(x,v),e^{\la t}E(x),e^{\la t}B(x)\right),
	\hspace{.25cm}
	\la>0.
	\end{equation}

The first step in proving these results is linearizing the problem (see \S\ref{equations1}). We then proceed by integrating Maxwell's Equations along the particle trajectories. The immediate benefit of this, is that we lose all dependence on the variable $v$, and can work, for the remainder of the paper, with functions of $x$ alone. These calculations appear in \S\ref{ml}, where operators that are closely related to $\Aonez$ and $\Atwoz$ naturally arise.

The integration along the particle paths leads us to three equations in three unknowns. Moreover, these equations depend upon the exponential growth parameter $\la$: finding a growing mode is equivalent to finding a nontrivial solution to this system for some $0<\la<\infty$. For convenience, we write the system as a matrix of operators, denoted by $\Ml$. We show that this matrix of operators is selfadjoint. Our general method for showing that the system has a solution for some $0<\la<\infty$ is to keep track of the spectrum of $\Ml$ as $\la$ varies, and show that it must cross $0$.

However, the matrix operator $\Ml$ has point spectrum extending both to $-\infty$ and to $+\infty$, due to the appearance of both $\p_x^2$ and $-\p_x^2$. Therefore, to keep track of the spectrum, in \S\ref{truncation} we truncate the matrix operator into a finite-dimensional mapping that we denote by $\Mnl$, where $n$ is the truncation parameter. A first obstacle is to show that for fixed $n$ the spectrum of $\Mnl$ varies in such a way that a nontrivial kernel is guaranteed for some $0<\la_n<\infty$. We do this by showing that $\Mnl$ has a different number of negative eigenvalues for small values of $\la$ and for large values of $\la$ and then applying the intermediate value theorem. The next difficulty lies in showing that the $\la_n$'s are uniformly bounded away from $0$ and from $+\infty$ for all $n$. We show this in Proposition \ref{lemma2} and in Lemma \ref{large}.

In \S\ref{sec5} we recover our original problem, by ``taking the limit $n\to\infty$.'' This process is tricky in itself, since we must define a proper sense in which finite-dimensional mappings ``tend'' to an infinite-dimensional mapping on Hilbert spaces. This is discussed, in part, in Lemma \ref{eq-converge} appearing in the preceding section.

Once we have shown that $\Ml$ has a nontrivial kernel for some $0<\la<\infty$, in \S\ref{construction} we demonstrate that the resulting solution is indeed a solution to the linearized RVM system. \S\ref{the-operators} contains many of the technical details used throughout the paper, whereas in \S\ref{sec-examples} we demonstrate our results in two examples.

The main difficulty of this paper is in the truncation, and in recovering the original problem from the truncated one. Moreover, the truncation is performed by means of a spectral projection onto finitely many eigenvalues of $\Mz$. However, the spectrum of $\Ml$ may change quite dramatically at $\la=0$, since all operators only converge strongly as $\la\to0$, but not in operator norm. Therefore, an additional difficulty is the need to deal with two varying parameters simultaneously -- the growth parameter $\la$ and the dimension of the truncation $n$.

Our result extends the results of \cite{rvm1}, and, more specifically, \cite{rvm2}, where Lin and Strauss established a linear stability criterion for equilibrium solutions that are \emph{strictly monotone}, that is, the number of particles at a given energy strictly decreases with the energy. They define a selfadjoint operator $\Lz$, acting on functions of $x$, with a rather simple spectral structure, and derive stability conditions based on those spectral properties. Namely, they prove that $\Lz\geq0$ implies spectral stability (in \cite{rvm1}), while $\Lz\ngeq0$ implies that there exists a growing mode (in \cite{rvm2}). The monotonicity assumption, which we drop in this paper, proves to be important in \cite{rvm1,rvm2}, as it determines the signs of some of the operators. This implies that $\Lz$ is not relevant to our case. The root cause of most of the analytical difficulties of our paper is the inability to determine the signs of any of the operators.\\

This paper is organized as follows:

\textbf{1. Introduction}.

\textbf{2. Setup}. We present our function spaces and operators.

\textbf{3. Behavior For Small and Large $\la$}. We ``count" eigenvalues for small and large values of $\la$.

\textbf{4. Limit as $n\to\infty$}. We retrieve our original system from the finite-dimensional approximations.

\textbf{5. Construction of a Growing Mode}. We verify that indeed the system retrieved is the right one.

\textbf{6. Examples}. We give two examples that illustrate our results.

\textbf{7. The Operators}. We discuss the main properties of the operators.

\section{Setup}

\subsection{The Function Spaces}
The function spaces we use throughout this paper are as follows: We denote by $\lpt$ and $H_P^k$ the usual Sobolev spaces of $P$-periodic functions on $\R$, that are square integrable on the interval $[0,P]$, as well as their first $k$ derivatives. In addition, we denote:

	\begin{eqnarray*}
	\lptz&=&\left\{h(x)\in \lpt\Lst\int_0^Ph\;dx=0\right\}\\
	H_{P,0}^k&=&\left\{h(x)\in H^k_P\Lst\int_0^Ph\;dx=0\right\}\\
	\lwt&=&\left\{m(x,v)\Lst m \text{ is } P\text{-periodic in }x, \left\|m\right\|_w^2:=\int_0^P\int_{\R^2}|m|^2w\;dv\;dx<\infty\right\}\\
	\end{eqnarray*}
where $w$ is the weight defined in \eqref{weight}. We note that $\lwt$ is independent of the choice of sign $\pm$, since $w$ only depends upon $e$ which is independent of $\pm$.
In addition, we use the following notation:
	
	\begin{align*}
	\|\cdot\|_{\lpt}\text{ and }\left<\cdot,\cdot\right>_{\lpt}  \text{ denote the norm and inner product in }L^2_{P}\text{ respectively}\\
	\|\cdot\|_{w}\text{ and }\left<\cdot,\cdot\right>_{w}  \text{ denote the norm and inner product in }L^2_{w}\text{ respectively}\\
	\end{align*}

\subsection{The Basic Equations}\label{equations1}

%
%

In the $1\frac{1}{2}$ dimensional case the RVM system becomes a system of scalar equations:
		\begin{subequations}\label{main-eq}
		\begin{align}
		&\p_tf^\pm+\hat{v}_1\p_xf^\pm\pm(E_1+\hat{v}_2B)\p_{v_1}f^\pm\pm(E_2-\hat{v}_1B)\p_{v_2}f^\pm
		=
		0						\label{vlasov}\\
		&\p_tE_1
		=
		-j_1						\label{ampere1}\\
		&\p_tE_2+\p_xB
		=
		-j_2						\label{ampere2}\\
		&\p_tB
		=
		-\p_xE_2					\label{faraday}\\
		&\p_xE_1
		=
		n_0+\rho						\label{gauss}
		\end{align}
		\end{subequations}
where $\hat{v}=v/\left<v\right>$ and $\left<v\right>=\sqrt{1+|v|^2}$, and
	\begin{subequations}
	\begin{align}
	\rho
	&=
	\int\left(f^+-f^-\right)\;dv	\\
	j_i
	&=
	\int\hat{v}_i\left(f^+-f^-\right)\;dv,\hspace{.25cm}i=1,2
	\end{align}
	\end{subequations}
and the external fields are replaced by the constant external radiation $n_0$. From our definitions \eqref{potentials} of the electric and magnetic potentials, we get the second order ODEs for $\phi$ and $\psi$:

	\begin{equation}
	\p_x^2\phi
	=
	-\p_xE_1
	=
	-\rho
	=
	-\int\left(f^+-f^-\right)\;dv
	\end{equation}

	\begin{equation}\label{psi-ode}
	\p_x^2\psi
	=
	\p_xB
	=
	-j_2-\p_tE_2
	=
	-\int\hat{v}_2\left(f^+-f^-\right)\;dv-\p_tE_2.\\
	\end{equation}

The linearized Vlasov equation is
	\begin{equation}\label{linvlasov}
	\left(\p_t+D^\pm\right)f^\pm
	=
	\mp\mu^\pm_e\hat{v}_1E_1\pm\mu^\pm_p\hat{v}_1B\mp\left(\mu^\pm_e\hat{v}_2+\mu^\pm_p\right)E_2,
	\end{equation}
where
	$$
	D^\pm
	=
	\hat{v}_1\p_x\pm\left(E_1^0+\hat{v}_2B^0\right)\p_{v_1}\mp\hat{v}_1B^0\p_{v_2}.
	$$

It is important to note that the \emph{purely magnetic} assumption implies that $E_1^0\equiv0$, but the perturbation $E_1$ need not vanish.
Since we are looking for a (purely) growing mode with exponent $\la$ (see \eqref{purely}), we replace everywhere $\p_t$ by $\la$. Thus, our equations for the $P$-periodic electric and magnetic (perturbed) potentials $\phi$ and $\psi$ become
	$$
	B=\p_x\psi
	$$
along with	
	$$
	E_2=-\la\psi,
	$$
which is a result of the integration of \eqref{faraday} and setting the constant of integration to be $0$, and, finally,
	$$
	E_1
	=
	-\p_x\phi-\la b,
	$$
where $b\in\R$ is meant to allow $E_1$ to have a nonzero mean.
%
Replacing $\p_t$ by $\la$, the linearized Vlasov equation becomes
	\begin{equation}\label{vlas}
	\left(\la+D^\pm\right)f^\pm
	=
	\pm\mu^\pm_e\hat{v}_1(\p_x\phi+\la b)\pm\mu^\pm_p\hat{v}_1\p_x\psi\pm\la\left(\mu^\pm_e\hat{v}_2+\mu^\pm_p\right)\psi,
	\end{equation}
along with Maxwell's Equations
	\begin{subequations}\label{max-eq}
	\begin{align}
	&-\la\p_x\phi-\la^2b=\la E_1=-j_1\label{ampere1d1}\\
	&-\la^2\psi+\p_x^2\psi=\la E_2+\p_xB=-j_2\label{ampere1d2}\\
	&-\p_x^2\phi=\p_xE_1=\rho.\label{gauss1d}
	\end{align}
	\end{subequations}

We see that in these equations there is only dependence upon derivatives of the electric potential $\phi$, and never dependence upon $\phi$ itself. Therefore, throughout this paper we let $\phi\in\lptz$. Now we introduce the particle paths $\left(X^\pm(s;x,v),V^\pm(s;x,v)\right)$ of the equilibrium state, governed by the transport operators $D^\pm$, where $-\infty<s<\infty$. They satisfy the system of ordinary differential equations
	\begin{subequations}
	\begin{align}
	\dot{X}^\pm&=\hat{V}^\pm_1\\
	\dot{V}^\pm_1&=
	\pm\hat{V}^\pm_2B^0(X^\pm)\\
	\dot{V}^\pm_2&=\mp\hat{V}^\pm_1B^0(X^\pm),
	\end{align}
	\end{subequations}
where $\dot{\square}=\p/\p s$ is the derivative along the characteristic curves, and the initial conditions are
	\begin{equation}
	\left(X^\pm(0,x,v),V^\pm(0,x,v)\right)
	=
	(x,v).
	\end{equation}
When there is no risk of confusion, we abbreviate $X^\pm=X^\pm(s)=X^\pm(s;x,v)$ and $V^\pm=V^\pm(s)=(V_1^\pm(s;x,v),V_2^\pm(s;x,v))$. Now we rewrite the Vlasov Equation integrated along the particle paths. Here it is crucial that $e^\pm$ and $p^\pm$, and any function of these variables, are constant along the trajectories. This implies that, $\mu_e^\pm$ and $\mu_p^\pm$ are constants under $s$-differentiation. We multiply \eqref{vlas} by $e^{\la s}$ and notice that the left hand side becomes the perfect derivative $\frac{\p}{\p s}\left(e^{\la s}f\right)$. Integrating the right hand side along the particle paths, one has
	\begin{align*}
	&\pm \int_{-\infty}^0e^{\la s}\left(\mu^\pm_e\hat{V}^\pm_1(\p_x\phi(X^\pm)+\la b)+\mu^\pm_p\hat{V}^\pm_1\p_x\psi(X^\pm)+\la\left(\mu^\pm_e\hat{V}^\pm_2+\mu^\pm_p\right)\psi(X^\pm)\right)ds\\
	=&\pm \int_{-\infty}^0e^{\la s}\mu^\pm_e\left(\hat{V}^\pm_1\p_x\phi(X^\pm)+\la\phi(X^\pm)\right) ds\mp \int_{-\infty}^0\la e^{\la s}\mu^\pm_e\phi(X^\pm)\;ds\\
	&\pm \int_{-\infty}^0e^{\la s}\mu^\pm_p\left(\hat{V}^\pm_1\p_x\psi(X^\pm)+\la\psi(X^\pm)\right)ds\\
	&\pm \int_{-\infty}^0\la e^{\la s}\mu_e^\pm\left(\hat{V}^\pm_1b+\hat{V}^\pm_2\psi(X^\pm)\right)ds=I+II+III+IV.
	\end{align*}
Recalling that $D^+$ and $D^-$ both reduce to $\hat{v}_1\p_x$ when applied to functions of $x$ alone (and not $v$), we see that the integrands in terms $I$ and $III$ are $e^{\la s}\mu_e^\pm\left(D\phi+\la\phi\right)$ and $e^{\la s}\mu_e^\pm\left(D\psi+\la\psi\right)$, respectively, evaluated along the appropriate particle paths. Therefore, both integrands become no more than $\frac{d}{ds}\left(e^{\la s}\mu_e^\pm\phi\right)$ and $\frac{d}{ds}\left(e^{\la s}\mu_e^\pm\psi\right)$. We conclude that the terms $I$ and $III$ become $\pm e^{\la s}\mu_e^\pm\phi(x)$ and $\pm e^{\la s}\mu_p^\pm\psi(x)$, with no boundary terms due to our decay assumptions. The other terms are kept in integral form as above. Since $\mu_e^\pm$ are constant along the trajectories, we may evaluate them at $s=0$, and they have no role in the integration. After dividing both sides by the exponent, we finally have
	\begin{align}\label{fpm}
	f^\pm(x,v)=&\pm\mu^\pm_e\phi(x)\pm\mu^\pm_p\psi(x)\\
	&\mp\mu^\pm_e\int_{-\infty}^0\la e^{\la s}\left[\phi(X^\pm(s))-\hat{V}^\pm_2(s)\psi(X^\pm(s))-b\hat{V}^\pm_1(s)\right]\;ds.\nonumber
	\end{align}

We simplify this expression by introducing the operators $\Ql_\pm: L^2_w\to L^2_w$, defined by:
	$$\left(\Ql_\pm k\right)(x,v)=\int_{-\infty}^0\la e^{\la s}k\left(X^\pm(s;x,v),V^\pm(s;x,v)\right)\;ds$$
where $k=k(x,v):\domain\to\R$.

\begin{rek}\label{ql-domain}
We note that if $h(x)\in\lpt$, then $\tilde{h}(x,v):=h(x), \;\;(x,v)\in[0,P]\times\R^2$, is clearly in $L^2_w$. Therefore, $\Ql_\pm$ act on functions in $\lpt$ as well. 
\end{rek}

With our definition of $\Ql_\pm$, \eqref{fpm} becomes

	\begin{equation}\label{partdist1d}
	\begin{array}{lcl}
	f^+(x,v)
	&=&
	+\mu^+_e\phi(x)+\mu^+_p\psi(x)-\mu^+_e\Ql_+\phi+\mu^+_e\Ql_+(\hat{v}_2\psi)+b\mu^+_e\Ql_+\hat{v}_1\\
	\\
	f^-(x,v)
	&=&
	-\mu^-_e\phi(x)-\mu^-_p\psi(x)+\mu^-_e\Ql_-\phi-\mu^-_e\Ql_-(\hat{v}_2\psi)-b\mu^-_e\Ql_-\hat{v}_1.
	\end{array}
	\end{equation}

An important quantity which we use frequently, is $f^+-f^-$. We write it explicitly for future reference:

	\begin{equation}\label{fplusfminus}
	f^+-f^-
	=
	\sum_\pm\left(\mu^\pm_e\phi(x)+\mu^\pm_p\psi(x)-\mu^\pm_e\Ql_\pm\phi+\mu^\pm_e\Ql_\pm(\hat{v}_2\psi)+b\mu^\pm_e\Ql_\pm\hat{v}_1\right)
	\end{equation}

We make the following two useful observations:
\begin{enumerate}
\item
It holds that
	\begin{equation}\label{fact}
	\int(\mu^\pm_p+\hat{v}_2\mu^\pm_e)\;dv=\int\frac{\p\mu^\pm}{\p v_2}\;dv=0.
	\end{equation}
\item
The particle paths preserve volume. That is, for fixed $s$,
	\begin{equation}\label{jaco}
	(X^\pm(s;x,v),V^\pm(s;x,v))\to(x,v)\text{ both have Jacobian}=1.
	\end{equation}
\end{enumerate}

\begin{lem}[Properties of $D^\pm$]
$D^\pm$ are skew-adjoint operators on $\lwt$. Their null spaces $\ker{D^\pm}$ consist of all functions $g=g(x,v)$ in $\lwt$ that are constant on each connected component in $\R\times\R^2$ of $\{e=const\text{ and }p^\pm=const\}$. In particular, $\ker{D^\pm}$ contain all functions of $e$ and of $p^\pm$.
\end{lem}

\begin{proof}
We show for the `$+$' case, and drop the $+$ superscripts. It is straightforward to verify that $De=Dp=0$. Therefore $\ker{D}$ contains all functions of $e$ and of $p$. Skew-adjointness is easily seen due to integration by parts, as $D$ is a first-order differential operator. Derivatives that ``hit" $w$ vanish, since $w=w(e)$ is a function of $e$.
\end{proof}

\begin{defi}
We define $\Proj^\pm$ to be the orthogonal projection operators of $L_w^2$ onto $\ker D^\pm$.
\end{defi}

\begin{lem}\label{proj-parity}
The projection operators $\Proj^\pm$ preserve parity with respect to the variable $v_1$.
\end{lem}

\begin{proof}
Let us demonstrate for $\Proj^+$ and drop the $+$ superscript to simplify notation. The demonstration for $\Proj^-$ is identical. Recall that
	$$
	D=D^+=\hat{v}_1\p_x+
	\hat{v}_2B^0\p_{v_1}-\hat{v}_1B^0\p_{v_2}.
	$$
Now, let $f=f(x,v_1,v_2)$ and let $R$ be the operator that reverses $v_1$: $Rh(x,v_1,v_2)=h(x,-v_1,v_2)$. Then
	$$
	D(Rf)
	=
	-R(\hat{v}_1\p_xf)-R(
	\hat{v}_2B^0\p_{v_1}f)+R(\hat{v}_1B^0\p_{v_2}f)
	=
	-R(Df)
	$$
Therefore $f\in\ker D$ if and only if $Rf\in\ker D$. This implies that one can find a basis of even and odd functions (in the variable $v_1$) to the space $\ker D$. To show that if $f$ is even (odd) in $v_1$ then $\Proj f$ is also even (odd) in $v_1$, we let $g\in \ker D$ be, without loss of generality, even or odd. Then it must hold that $\iint(f-\Proj f)g\;w\;dx\;dv=0$. In the case that $f$ is even, we change variables $v_1\to-v_1$ to get $\pm\iint(f-R(\Proj f))g\;w\;dx\;dv=0$ and, therefore, $R(\Proj f)=\Proj f$. Here the $\pm$ depends on the parity of $g$. The odd case is treated in the same way.
\end{proof}

	\begin{lem}[Properties of $\Ql_\pm$]\label{propql}
	Let $0<\la<\infty$.
		\begin{enumerate}
		\item\label{qlnorm}

		$\Ql_\pm$ map $\lwt\to\lwt$ with operator norm = 1. Moreover, recalling Remark \ref{ql-domain}, $\Ql_\pm$ are also bounded as operators $\lpt\to\lwt$.

		\item\label{laz}
	For all $m(x,v)\in\lwt$, $\left\|\Ql_\pm m-\Proj^\pm m\right\|_w\to0$ as $\la\to0$.

		\item\label{lanotz}
		If $\sigma>0$, then $\left\|\Ql_\pm-\Qs_\pm\right\|=O(|\la-\sigma|)$ as $\la\to\sigma$, where $\|\cdot\|$ is the operator norm from $\lwt$ to $\lwt$.

		\item\label{lainfty}
		For all $m(x,v)\in\lwt$, $\left\|\Ql_\pm m-m\right\|_w\to0$ as $\la\to\infty$.
		
		\item\label{almost-sym}
		Let $\tilde{v}=(-v_1,v_2)$ and let $\tilde{n}(x,v)=n(x,\tilde{v})$. Then $\left<\Ql_\pm m,n\right>_w=\left<m,\Ql_\pm\tilde{n}\right>_w$.
		\end{enumerate}
	\end{lem}

\begin{proof}
Let us show for the `$+$' case, and drop the $+$ indices to simplify notation.
\begin{enumerate}
\item
We demonstrate showing the two bounds on $\Ql$ using a dual method and using a direct method; this shows that either method works, and each has its benefits. Let $m,n\in\lwt$.
	$$
	\left<\Ql m,n\right>_w
	=
	\int_{-\infty}^0\la e^{\la s}\int_0^P\int(m\sqrt{w})\left(X(s),V(s)\right)\cdot (n\sqrt{w})(x,v)\;dv\;dx\;ds
	\leq
	\|m\|_w\|n\|_w,
	$$
where in the equality we used the fact that $w=w(e)$ is constant along the trajectories (this is true since $e$ is constant along the trajectories), and the inequality is simply the Cauchy-Schwarz Inequality. The assertion that the operator norm is $1$ is verified since $\Ql1=1$ (note that $1\in L_w^2$).

Now, for the $\lpt$ bound we let $h(x)\in\lpt$ and use \eqref{jaco} to get
	\begin{align}
	\left\|\Ql h\right\|_w^2
	&=
	\int_0^P\int\left|\int_{-\infty}^0\la e^{\la s}h(X(s;x,v))\;ds\right|^2w\;dv\;dx\nonumber\\
	&\leq
	\int_0^P\int\left(\int_{-\infty}^0 \left|\la e^{\la s}h(x)\right|ds\right)^2w\;dv\;dx\label{ql-estimate}\\
	&=
	\int_0^P\int w\;dv\;|h(x)|^2\;dx\nonumber\\
	&\leq
	\left(\sup_x\int w\;dv\right)\left\|h\right\|_{\lpt}^2
	=
	C\left\|h\right\|^2_{\lpt}.\nonumber
	\end{align}

\item
We let $M$ denote the spectral measure of the selfadjoint operator $T=-iD$ in the space $\lwt$. We can therefore write, for any $m=m(x,v)$, $m(X(s),V(s))=e^{sD}m=e^{isT}m$, so that
$$
\Ql m=\int_{-\infty}^0\la e^{\la s}m(X(s),V(s))\;ds=\int_{-\infty}^0\la e^{\la s}\int_{\R}e^{i\alpha s}dM(\alpha)m\;ds=\int_{\R}\frac{\la}{\la+i\alpha}\;dM(\alpha)m.
$$

The projection $\Proj$ can be written as $\Proj=M(\{0\})=\int_{\R}\chi\;dM$ where $\chi(0)=1$ and $\chi(\alpha)=0$ for $\alpha\neq0$. Thus
$$
\left\|\Ql m-\Proj m\right\|_{w}^2=\left\|\int_{-\infty}^0\la e^{\la s}m(X(s),V(s))\;ds-\Proj m\right\|_{w}^2=\int_{\R}\left|\frac{\la}{\la+i\alpha}-\chi(\alpha)\right|^2\;d\left\|M(\alpha)m\right\|_{w}^2
$$
due to the orthogonality of spectral projections. We use the dominated convergence theorem to finish the proof that, indeed, this expression tends to $0$ as $\la\to0$.

\item
To show that $\left\|\Ql_\pm-\Qs_\pm\right\|=O(|\la-\sigma|)$ as $\la\to\sigma$ we use the fact that $(x,v)\to(X,V)$ has Jacobian$=1$:
\begin{eqnarray*}
\|\Ql m-\Qs m\|_w
&\leq&
\int_{-\infty}^0\left|\la e^{\la s}-\sigma e^{\sigma s}\right|\left\|m(X(s),V(s))\right\|_wds\\
&=&
\int_{-\infty}^0\left|\la e^{\la s}-\sigma e^{\sigma s}\right|ds\|m\|_w\\
&\leq&
C|\ln{\la}-\ln{\sigma}|\;\|m\|_w.
\end{eqnarray*}

\item
With $M$ the same spectral measure as above, we have:
\begin{equation*}
\Ql m-m=\int_{\R}\left(\frac{\la}{\la+i\alpha}-1\right)dM(\alpha)m.
\end{equation*}
Therefore
\begin{equation*}
\|\Ql m-m\|^2_w\leq\int_{\R}\left|\frac{\la}{\la+i\alpha}-1\right|^2d\|M(\alpha)m\|^2_w\to0
\end{equation*}
as $\la\to\infty$, by the dominated convergence theorem.

\item
	\begin{eqnarray*}
	\left<\Ql m,n\right>_w
	&=&
	\int_0^P\int\Ql (m(x,v))\;n(x,v)\;w\;dv\;dx\\
	&=&
	\int_0^P \int\left(\int_{-\infty}^0\la e^{\la s}m(X(s),V(s))\;ds\right)n(x,v)\;w\;dv\;dx\\
	&=&
	\int_0^P \int\int_{-\infty}^0\la e^{\la s}m(x,v)\;n(X(-s),V(-s))\;w\;ds\;dv\;dx\\
	&=&
	\int_0^P \int\left(\int_{-\infty}^0\la e^{\la s}\tilde{n}(X(s),V(s))\;ds\right)m(x,v)\;w\;dv\;dx\\
	&=&
	\left<m,\Ql\tilde{n}\right>_w,
	\end{eqnarray*}
where for the third equality we used \eqref{jaco} and the fact that $w$ is invariant under $D$, and for the fourth equality we used the fact that
	\begin{equation}\label{coor-change}
	\begin{array}{rcl}
	X(-s;x,-v_1,v_2)&=&X(s;x,v_1,v_2)\\
	-V_1(-s;x,-v_1,v_2)&=&V_1(s;x,v_1,v_2)\\
	V_2(-s;x,-v_1,v_2)&=&V_2(s;x,v_1,v_2).
	\end{array}
	\end{equation}

This type of calculation appears often, so it is worthwhile writing in detail.
\end{enumerate}
\end{proof}

\subsection{The Operators}

In addition to the definitions of $\Aonez$ and $\Atwoz$ in \S\ref{mainres}, we define the following operators depending on a real parameter $0<\lambda<\infty$, acting on $\lptz,\lpt,\lpt$, with domains $\hptz,\hpt,\lpt$, respectively:
	\begin{eqnarray*}
	\Aonel h&=&-\p_x^2h-\left(\sum_\pm\int\mu^\pm_e\;dv\right)h+\sum_\pm\int\mu^\pm_e\Ql_\pm h\;dv,\\
	\Atwol h&=&-\p_x^2h+\la^2h-\left(\sum_\pm\int\hat{v}_2\mu^\pm_p\;dv\right)h-\sum_\pm\int\mu^\pm_e\hat{v}_2\Ql_\pm (\hat{v}_2h)\;dv,\\
	\Bl h&=&\left(\sum_\pm\int\mu^\pm_p\;dv\right)h+\sum_\pm\int\mu^\pm_e\Ql_\pm(\hat{v}_2h)\;dv.
	\end{eqnarray*}

As we have seen in Lemma \ref{propql}, $\Ql_{\pm}\to\Proj^\pm$ strongly as $\la\to0$. We also define the following multiplication operators with domain $\R$ and range $\lpt$, depending on the parameter $0<\la<\infty$:
	\begin{eqnarray*}
	\Cl(b)&=&b\sum_\pm\int\mu^\pm_e\Ql_\pm\left(\hat{v}_1\right)dv,\\
	\Dl(b)&=&b\sum_\pm\int\hat{v}_2\mu^\pm_e\Ql_\pm\left(\hat{v}_1\right)dv,\\
	\end{eqnarray*}
and a constant depending on $\la$:
	$$
	l^\la=\frac{1}{P}\sum_\pm\int_0^P\int\hat{v}_1\mu^\pm_e\Ql_\pm\left(\hat{v}_1\right)dv\;dx,
	$$
where, for $\la=0$ we define
	$$
	l^0=\frac{1}{P}\sum_\pm\int_0^P\int \hat{v}_1\mu^\pm_e\Proj^\pm\left(\hat{v}_1\right)dv\;dx.
	$$

Finally, we derive formulas for the adjoint operators of $\Bl$, $\Cl$ and $\Dl$. We begin with $\left(\Bl\right)^*$. Let $h,k\in\lpt$. To simplify notation in this calculation, we drop the summation over $\pm$. All calculations work similarly with the proper definition that includes the $\pm$.
	
	\begin{eqnarray*}
	\left<\Bl h,k\right>_{\lpt}
	&=&
	\int_0^P\left[\left(\int\mu_p\;dv\right)h(x)+\int\mu_e\Ql(\hat{v}_2h)\;dv\right]\;k(x)\;dx\\
	&=&
	\int_0^P\left[\left(\int\mu_p\;dv\right)h(x)+\int\mu_e\left(\int_{-\infty}^0\la e^{\la s}\hat{V}_2(s)h(X(s))\;ds\right)\;dv\right]\;k(x)\;dx\\
	&=&
	\int_0^P\left[\left(\int\mu_p\;dv\right)h(x)k(x)+\int\mu_e\left(\int_{-\infty}^0\la e^{\la s}\hat{v}_2h(x)k(X(-s))\;ds\right)\;dv\right]\;dx\\
	&=&
	\int_0^P\left[\left(\int\mu_p\;dv\right)k(x)+\int\mu_e\hat{v}_2\left(\int_{-\infty}^0\la e^{\la s}k(X(s))\;ds\right)\;dv\right]\;h(x)\;dx\\
	\end{eqnarray*}

Therefore (with the $\pm$)

	$$
	\left(\Bl\right)^*k=\left(\sum_\pm\int\mu^\pm_p\;dv\right)k+\sum_\pm\int\mu^\pm_e\hat{v}_2\Ql_\pm k\;dv.
	$$
	
The computation of $\left(\Cl\right)^*$ is much simpler. Let $b\in\R$ and $k\in\lpt$. One has

	$$
	\left<\Cl b,k\right>_{\lpt}
	=
	b\int_0^P\left(\sum_\pm\int\mu^\pm_e\Ql_\pm\left(\hat{v}_1\right)\;dv\right)k(x)\;dx,
	$$
and, therefore
	$$
	\left(\Cl\right)^*k=\sum_\pm\int_0^P\int\mu^\pm_e\Ql_\pm\left(\hat{v}_1\right)k(x)\;dv\;dx.
	$$
The derivation of $\left(\Dl\right)^*$ is similar and yields
	$$
	\left(\Dl\right)^*k=\sum_\pm\int_0^P\int\hat{v}_2\mu^\pm_e\Ql_\pm\left(\hat{v}_1\right)k(x)\;dv\;dx.
	$$

\subsection{The Matrix Operator $\Ml$}\label{ml}
In this section we rewrite Maxwell's Equations \eqref{max-eq} in terms of the the unknowns $\phi, \psi$ and $b$. The various operators acting on these unknowns are precisely $\Aonel, \Atwol, \Bl, \Cl$ and $l^\la$. We show that Maxwell's Equations reduce to a simple, selfadjoint matrix operator $\Ml$, and that there exists a nontrivial solution to Maxwell's Equations if and only if $\Ml$ has a nontrivial kernel. Note that in what follows, the dependence upon $f^\pm$ enters  through the right-hand-side of \eqref{max-eq}.
\begin{enumerate}

\item
Rewriting Gauss' Equation \eqref{gauss1d} as $\p_x^2\phi+\rho=0$, we substitute \eqref{fplusfminus}, to get
	\begin{eqnarray*}
	0
	&=&
	\p_x^2\phi+\rho=\p_x^2\phi+\int(f^+-f^-)\;dv\\
	&=&
	\p_x^2\phi+\sum_\pm\int\left(\mu^\pm_e\phi(x)+\mu^\pm_p\psi(x)-\mu^\pm_e\Ql_\pm\phi+\mu^\pm_e\Ql_\pm(\hat{v}_2\psi)+ b\mu^\pm_e\Ql_\pm\hat{v}_1\right)\;dv\\
	&=&
	\p_x^2\phi+\left(\sum_\pm\int\mu^\pm_e\;dv\right)\phi(x)-\sum_\pm\int\mu^\pm_e\Ql_\pm\phi\;dv\\&&+\left(\sum_\pm\int\mu^\pm_p\;dv\right)\psi(x)+\sum_\pm\int\mu^\pm_e\Ql_\pm(\hat{v}_2\psi)\;dv+b\sum_\pm\int\mu^\pm_e\Ql_\pm\hat{v}_1\;dv\\
	&=&
	-\Aonel\phi+\Bl\psi+\Cl b.
	\end{eqnarray*}

Thus, our first relation is
	\begin{equation}\label{firsteq}
	-\Aonel\phi+\Bl\psi+\Cl b=0.
	\end{equation}

\item
Rewriting the second of Amp\`{e}re's Equations \eqref{ampere1d2} as $\p_x^2\psi-\la^2\psi+j_2=0$ and repeating the same procedure, we have
	\begin{eqnarray*}
	0
	&=&
	\p_x^2\psi-\la^2\psi+j_2\\
	&=&
	\p_x^2\psi-\la^2\psi+\sum_\pm\int\hat{v}_2\left(\mu^\pm_e\phi(x)+\mu^\pm_p\psi(x)-\mu^\pm_e\Ql_\pm\phi+\mu^\pm_e\Ql_\pm(\hat{v}_2\psi)+b\mu^\pm_e\Ql_\pm\hat{v}_1\right)\;dv\\
	&=&
	\p_x^2\psi-\la^2\psi+\left(\sum_\pm\int\hat{v}_2\mu^\pm_p\;dv\right)\psi+\sum_\pm\int\hat{v}_2\mu^\pm_e\Ql_\pm(\hat{v}_2\psi)\;dv\\&&+\left(\sum_\pm\int\hat{v}_2\mu^\pm_e\;dv\right)\phi(x)-\sum_\pm\int\hat{v}_2\mu^\pm_e\Ql_\pm\phi\;dv+b\sum_\pm\int\hat{v}_2\mu^\pm_e\Ql_\pm\hat{v}_1\;dv\\
	&=&
	-\Atwol\psi-\left(\Bl\right)^*\phi+\Dl b.
	\end{eqnarray*}
where in the last equality we used \eqref{fact}. Our second result is
	\begin{equation}\label{secondeq}
	\left(\Bl\right)^*\phi+\Atwol\psi-\Dl b=0.
	\end{equation}

\item
Finally we consider the first of Amp\`{e}re's Equations \eqref{ampere1d1}: $0=\la E_1+j_1=-\la\p_x\phi-\la^2b+j_1$. Integrating over one period, we get
	$$
	\la^2b=\frac{1}{P}\int_0^Pj_1\;dx.
	$$

Plugging in the expression for the particle density $f^+-f^-$ as before, we have

	\begin{eqnarray*}
	\la^2b
	&=&
	\frac{1}{P}\int_0^Pj_1\;dx=\frac{1}{P}\int_0^P\int\hat{v}_1(f^+-f^-)\;dv\;dx\\
	&=&
	\frac{1}{P}\sum_\pm\int_0^P\int\hat{v}_1\left(\mu^\pm_e\phi(x)+\mu^\pm_p\psi(x)-\mu^\pm_e\Ql_\pm\phi+\mu^\pm_e\Ql_\pm(\hat{v}_2\psi)+b\mu^\pm_e\Ql_\pm\hat{v}_1\right)dv\;dx\\
	\end{eqnarray*}

The first two terms vanish since $\mu$ is even in $v_1$. Thus we are left with
	
	$$
	\la^2b=\frac{1}{P}\sum_\pm\int_0^P\int\hat{v}_1\mu^\pm_e\left(-\Ql_\pm\phi+\Ql_\pm(\hat{v}_2\psi)+b\Ql_\pm\hat{v}_1\right)\;dv\;dx.
	$$

Denoting these three terms $I,II,III$ respectively, we find that
	\begin{eqnarray*}
	I
	&=&
	-\frac{1}{P}\sum_\pm\int_0^P\int\hat{v}_1\mu^\pm_e\Ql_\pm\phi\;dv\;dx\\
	&=&
	-\frac{1}{P}\sum_\pm\int_0^P \int\int_{-\infty}^0\la e^{\la s}\hat{v}_1\mu^\pm_e\phi(X^\pm(s))\;ds\;dv\;dx\\
	&=&
	-\frac{1}{P}\sum_\pm\int_0^P \int\int_{-\infty}^0\la e^{\la s}\hat{V}^\pm_1(-s)\mu^\pm_e\phi(x)\;ds\;dv\;dx\\
	&=&
	\frac{1}{P}\sum_\pm\int_0^P \int\int_{-\infty}^0\la e^{\la s}\hat{V}^\pm_1(s)\mu^\pm_e\phi(x)\;ds\;dv\;dx\\
	&=&
	\frac{1}{P}\left(\Cl\right)^*\phi,
	\end{eqnarray*}
where for the third equality we changed variables $(x,v)\to(X(-s),V(-s))$ and used the facts that this change of variables has Jacobian =1 and that $\mu$ is invariant under $D$, and for the fourth equality we changed the variable $v\to-v$ and used the change of coordinates prescribed in \eqref{coor-change}. Similarly, $$II=-\frac{1}{P}\left(\Dl\right)^*\psi,$$ and, finally, $$III=bl^\la$$ by definition. Summarizing, we have
	\begin{equation}\label{thirdeq}
	\left(\Cl\right)^*\phi-\left(\Dl\right)^*\psi-Pb\left(\la^2-l^\la\right)=0.
	\end{equation}

\end{enumerate}
Motivated by the three equations \eqref{firsteq}-\eqref{thirdeq} depending upon the parameter $\la$, we define the matrix operator $\Ml:\lpt\times\lpt\times\R\to \lpt \times \lpt \times \R$
	$$
	\Ml
	=
	\matthree{-\Aonel}{\Bl}{\Cl}{\left(\Bl\right)^*}{\Atwol}{-\Dl}{\left(\Cl\right)^*}{-\left(\Dl\right)^*}{-P\left(\la^2-l^\la\right)}.
	$$
with domain $\hpt\times\hpt\times\R$.
Formally, to prove our main theorem, it suffices to show that $\Ml$ has a nontrivial kernel for some $0<\la<\infty$.
Accordingly, we also define
	$$
	\Mz
	=
	\matthree{-\Aonez}{0}{0}{0}{\Atwoz}{0}{0}{0}{Pl^0}.
	$$

\begin{rek}\label{remark}
As mentioned before, since $\phi$ only matters up to a constant, we restrict the domain of $\Ml$ and $\Mz$ to $\hptz\times\hpt\times\R$. Indeed, making this restriction is important. Letting $(\phi,\psi,b)=u_{triv}^T=(1,0,0)$ we notice that $\Ml u_{triv}=0$ for any $\la\geq0$. However, $u_{triv}$ is a trivial solution that is of no interest for us, since it would generate a trivial solution $(f,E,B)=(0,0,0)$. Moreover, multiples of $u_{triv}$ are the only trivial solutions. Indeed, whenever either $\psi$ or $b$ are nonzero, the linearized equations become nontrivial.
\end{rek}

\section{Behavior for Small and Large $\lambda$}\label{truncation}
This section is the heart of this paper. To show that $\Ml$ has a nontrivial kernel for some $0<\la<\infty$, we wish to show that there is some eigenvalue that crosses $0$ as $\la$ increases, starting at $0$. To achieve this, we must somehow be able to keep track of the spectrum. Since we are working on the finite interval $[0,P]$ we do not have to worry about a continuous spectrum. However, since $\Aonel$ and $\Atwol$ appear with opposite signs in $\Ml$, we expect to find eigenvalues near $-\infty$ as well as near $+\infty$. Thus, to facilitate the counting, we truncate the spectrum, and leave only a finite number of eigenvalues. The following projection is defined so as to preserve the spectral properties of the two operators appearing along the diagonal: $\Aonez$ and $\Atwoz$.

We postpone discussing the properties of the various aforementioned operators until \S\ref{the-operators}. However, these properties are used throughout this section.



%
Consider the eigenvalues $\alpha_1\leq\alpha_2\leq\cdots$ of $\Aonez$ in ascending order, counting multiplicity. Let $\xi_i$ be unit eigenvectors associated to $\alpha_i$ chosen to be mutually orthonormal.
We define the projection $P_n:\lptz\to\R^n$ to be:
	\begin{equation}
	P_nu
	=
	\left\{\left<u,\xi_i\right>_{\lpt}\right\}_{i=1}^n,
	\end{equation}
and, thus, $P_n^*:\R^n\to\lptz$ is given by
	\begin{equation}
	P_n^*a
	=
	\sum_{i=1}^na_i\xi_i,
	\end{equation}
where $a=(a_1,a_2,\dots,a_n)\in\R^n$.
Hence, 
	\begin{equation}
	P_n^*P_nu=\sum_{i=1}^n\left<u,\xi_i\right>_{\lpt}\xi_i
	\end{equation}
is the projection onto the eigenspace associated with the first $n$ eigenvalues of $\Aonez$.

Similarly, we define $Q_n$ to be the projection operator onto the eigenspace spanned by the $n$ eigenvectors $\zeta_1,\dots,\zeta_n$ of the operator $\Atwoz$ associated to its first $n$ eigenvalues $\beta_1\leq\beta_2\leq\cdots\leq\beta_n$.
We are now ready to define the approximate matrix operator $\Mnl$:
	\begin{equation}
	\Mnl
	=
	\matthree{-P_n\Aonel P_n^*}{P_n\Bl Q_n^*}{P_n\Cl}{Q_n\left(\Bl\right)^*P_n^*}{Q_n\Atwol Q_n^*}{-Q_n\Dl}{\left(\Cl\right)^*P_n^*}{-\left(\Dl\right)^*Q_n^*}{-P\left(\la^2-l^\la\right)}.
	\end{equation}
When $\la=0$, this definition reduces to
	\begin{equation}
	\Mnz
	=
	\matthree{-P_n\Aonez P_n^*}{0}{0}{0}{Q_n\Atwoz Q_n^*}{0}{0}{0}{Pl^0}.
	\end{equation}
Both matrices are finite-dimensional mappings $\R^n\times\R^n\times\R\to\R^n\times\R^n\times\R$. Let us list a few facts which will be useful for us later on:

\begin{enumerate}
\item
For all $\la\geq0$ and for any $n\in\N$, $\Mnl$ is selfadjoint.
\item
Let $\Sigma(\mathcal{A})$ denote the spectrum of the operator $\mathcal{A}$. Then $\Sigma(\Mnz)\subseteq\Sigma(\Mz)$.

\begin{proof}
This is clearly true due to the properties of $P_n$ and $Q_n$, and the diagonal structure of both matrix operators.
\end{proof}

\item
There exists some $\sigma^*<0$ such that $\left[\sigma^*,0\right)\cap\Sigma(\Mz)=\emptyset$.

\begin{proof}
Since the spectrum of $\Ml$ is discrete (with no finite accumulation points) for all $\la\geq0$, there exists some $\sigma_0<0$ that is the greatest negative eigenvalue of $\Mz$. We can choose $\sigma^*\in(\sigma_0,0)$ freely.
\end{proof}

\item
For fixed $n>0$, $\Mnl$ varies continuously in $\la$ as a mapping $\R^n\times\R^n\times\R\to\R^n\times\R^n\times\R$, for all $\la\geq0$.

\begin{proof}
For $\la>0$ this is clear, since $\Ml$ varies continuously by Lemma \ref{propml}. Thus we focus on the case $\la=0$. Since all norms on $\R^k$ are equivalent, it will suffice to check that $\Mnl\to\Mnz$ strongly as $\la\to0$. Let $u^T=(\phi,\psi,b)\in\R^n\times\R^n\times\R$. We need to show that $\|(\Mnl-\Mnz)u\|\to0$ as $\la\to0$. Since
	$$
	(\Mnl-\Mnz)u
	=
	\matthree{-P_n(\Aonel-\Aonez) P_n^*}{P_n\Bl Q_n^*}{P_n\Cl}{Q_n\left(\Bl\right)^*P_n^*}{Q_n(\Atwol-\Atwoz) Q_n^*}{-Q_n\Dl}{\left(\Cl\right)^*P_n^*}{-\left(\Dl\right)^*Q_n^*}{-P\left(\la^2-l^\la+l^0\right)}
	\left(\begin{array}{c}
	\phi\\\psi\\b
	\end{array}\right),
	$$
it suffices to show that terms of the form $|P_n(\Aonel-\Aonez) P_n^*\phi|$ or $|P_n\Bl Q_n^*\psi|$ tend to $0$ for fixed $n$, as $\la\to0$. This is clearly true due to the properties of the various operators discussed in \S\ref{the-operators}.
\end{proof}
\item
For fixed $n>0$, $P_n^*P_n$ is bounded both in $H^1_P$ and in $H^2_P$.
\begin{proof}
This follows from the fact that $P_n^*P_n$ maps $\lpt$ into $\hpt$ and the closed graph theorem.
\end{proof}
\end{enumerate}
\begin{lem}\label{conv-h2}
For any $g\in\hpt$, $P_n^*P_ng\to g$ in $H^1_P$ and in $\hpt$, as $n\to\infty$. Similarly, $Q_n^*Q_ng\to g$ in $H^1_P$ and in $\hpt$.
\end{lem}

\begin{proof}
We prove for $P_n^*P_n$; the proof for $Q_n^*Q_n$ is similar.
Assume that $g=\sum_{j=1}^\infty g_j\xi_j$ in $\lpt$, where $\xi_j$ are eigenvectors of $\Aonez$, as defined above. Since $g\in\hpt=D(\Aonez)$, we know that $\Aonez g\in\lpt$ and, therefore, there exist some $\beta_j$ such that $\Aonez g=\sum_{j=1}^\infty \beta_j\xi_j$ with $\sum_{j=1}^\infty\left|\beta_j\right|^2<\infty$. In fact, by taking the $\lpt$-inner product of $\Aonez g$ with $\xi_k$, we easily see that $\beta_k=g_k\alpha_k$, where $\alpha_k$ is the $k$th eigenvalue of $\Aonez$.

Using Poincar\'{e}'s Inequality (twice) and then the triangle inequality we have
	\begin{align*}
	\left\|\sum_{j=n}^\infty g_j\xi_j\right\|_{\hpt}
	\leq
	C\left\|\p_x^2\sum_{j=n}^\infty g_j\xi_j\right\|_{\lpt}
	&\leq
	C\left\|\left(\Aonez+\p_x^2\right)\sum_{j=n}^\infty g_j\xi_j\right\|_{\lpt}
	+
	C\left\|\Aonez\sum_{j=n}^\infty g_j\xi_j\right\|_{\lpt}
	\end{align*}
The first term tends to $0$ since $\Aonez+\p_x^2$ is a bounded operator on $\lpt$ (see Lemma \ref{properties2}\eqref{op-norm}), and, therefore, this term is controlled by $C'\left\|\sum_{j=n}^\infty g_j\xi_j\right\|^2$ which tends to $0$ as $n\to\infty$ since $g\in\lpt$. The second term tends to $0$ since $\Aonez\sum_{j=n}^\infty g_j\xi_j=\sum_{j=n}^\infty g_j\alpha_j\xi_j$, and, as mentioned above, $\sum_{j=1}^\infty\left|g_j\alpha_j\right|^2<\infty$. This shows strong convergence in $\hpt$. The strong convergence in $H^1_P$ is due to interpolation between $\lpt$ and $\hpt$.
\end{proof}

\begin{cor}
Considering the restrictions of $P_n^*P_n$ to $H^1_P$, the lemma implies by the uniform boundedness theorem that
	\be\label{uni-bound-proj}
	\sup_n\|P_n^*P_n\|_{H^1_P\to H^1_P}<\infty.
	\en
\end{cor}

\begin{lem}\label{h-1}
The operator $P_n^*P_n$ can be extended to $H_P^{-1}$. The sequence of extended operators $\{P_n^*P_n\}_{n=1}^\infty$ converges strongly to the identity.
\end{lem}

\begin{proof}
Our basic tool in extending the domain of definition is the canonical identification of $\lpt$ with a subspace of $H_P^{-1}$ via the scalar product. With that standard definition at hand, we can now extend $P_n^*P_n$ to $H_P^{-1}$. For brevity, denote $J_n=P_n^*P_n:\lpt\to\hpt$. As noted above, the restriction to $J_n:H^1_P\to H^1_P$ is bounded. Hence the dual operator $J_n^*:H^{-1}_P\to H^{-1}_P$ is bounded with the same operator norm (see \cite[III-\S3.3]{kato}). In fact, $J_n^*$ is an extension of $J_n$ to $H^{-1}_P$. Indeed, if $g\in\lpt$ and $f\in H^1_P$, and using the $\left<H^1_P,H^{-1}_P\right>$ paring, one has
	\be
	\left<f,J_n^*g\right>=\left<J_n f,g\right>=\left<J_n f,g\right>_{\lpt}=\left<f,J_n g\right>_{\lpt}.
	\en
Finally, we must show that the sequence $\{J_n^*\}_{n=1}^\infty$ converges strongly to the identity in $H^{-1}_P$ as $n\to\infty$. In view of \eqref{uni-bound-proj}, also $\sup_n\|J_n^*\|_{H^{-1}_P\to H^{-1}_P}<\infty$. Therefore, it suffices to prove that $J_n^*g\to g$ for $g$ in a dense subset of $H^{-1}_P$. Since we already know that strong convergence holds for $g\in\lpt$, we are done.
\end{proof}

\begin{lem}\label{eq-converge}
Let $u_n^T=(\phi_n,\psi_n,b_n)\in\R^n\times\R^n\times\R$, and suppose that $(P_n^*\phi_n,Q_n^*\psi_n,b_n)\to u_0^T=(\phi_0,\psi_0,b_0)$ strongly in $\lpt\times\lpt\times\R$ as $n\to\infty$ with $P_n^*\phi_n,Q_n^*\psi_n$ uniformly bounded in $\hpt$. Suppose that $\la_n\to\la_0\in[0,\infty)$ as $n\to\infty$. Then $\Mnmlnm u_n\to\Mlz u_0$ weakly, in the sense that
	$$
	\left<\Mnmlnm u_n,v_n\right>_{\R^n\times\R^n\times\R}
	\to
	\left<\Mlz u_0,v\right>_{\lpt\times\lpt\times\R}.
	$$
for any $v\in\lpt\times\lpt\times\R$ with $v_n\in\R^n\times\R^n\times\R$ being the projections of the two first corrdinates of $v$ and the identity on the third.
\end{lem}

\begin{proof}
Let us write $\Mnmlnm u_n^T$ and $\Mlz u_0$ explicitly:

	\begin{eqnarray}
	\Mnmlnm u_n
	&=&
	\matthree{-P_n\Aonelk P_n^*}{P_n\Blk Q_n^*}{P_n\Clk}{Q_n\left(\Blk\right)^*P_n^*}{Q_n\Atwolk Q_n^*}{-Q_n\Dlk}{\left(\Clk\right)^*P_n^*}{-\left(\Dlk\right)^*Q_n^*}{-P\left(\la_n^2-l^{\la_n}\right)}
	\left(\begin{array}{c}\phi_n\\\psi_n\\b_n\end{array}\right)\\
	\Mlz u_0
	&=&
	\matthree{-\Aonelz}{\Blz}{\Clz}{\left(\Blz\right)^*}{\Atwolz}{\Dlz}{\left(\Clz\right)^*}{\left(\Dlz\right)^*}{-P\left(\la^2_0-l^{\la_0}\right)}
	\left(\begin{array}{c}\phi_0\\\psi_0\\b_0\end{array}\right).
	\end{eqnarray}

We show the weak convergence term by term, by its location in the matrix. Let us demonstrate for the first term -- the rest are either similar, or simpler.

We want to show that
$
P_n\Aonelnm P_n^*\fnm \to\Aonelz\phiz
$
weakly. Since the sequence $\Aonelnm P_n^*\fnm$ is uniformly bounded in $\lpt$ by assumption, it suffices to test the weak continuity by taking $g\in\hpt$ which is dense in $\lpt$. We have
	\begin{eqnarray*}
	\left|\left<P_n\Aonelnm P_n^*\fnm,P_ng\right>_{\R^n}-\left<\Aonelz\phiz,g\right>_{\lpt}\right|
	&=&
	\left|\left<P_n^*P_n\Aonelnm P_n^*\fnm-\Aonelz\phiz,g\right>_{\lpt}\right|\\
	&\leq&
	\left|\left<P_n^*P_n\left(\Aonelnm -\Aonelz \right)P_n^*\fnm,g\right>_{\lpt}\right|\\
	&&+\left|\left<\left(P_n^*P_n-I\right)\Aonelz P_n^*\fnm,g\right>_{\lpt}\right|\\
	&&+\left|\left<\Aonelz \left(P_n^*\fnm-\phiz\right),g\right>_{\lpt}\right|\\
	&=&I+II+III.
	\end{eqnarray*}

Terms $I$ and $II$ are treated very similarly:
	$$
	I
	=
	\left|\left< P_n^*\fnm,\left(\Aonelnm-\Aonelz\right)P_n^*P_n g\right>_{\lpt}\right|\\
	\leq
	\left\|P_n^*\fnm\right\|_{\lpt}\left\|\left(\Aonelnm-\Aonelz\right)P_n^*P_n g\right\|_{\lpt}\\
	\to
	0,
	$$
	and
	$$
	II
	=
	\left|\left<P_n^*\fnm,\Aonelz \left(P_n^*P_n-I\right)g\right>_{\lpt}\right|\\
	\leq
	\left\|P_n^*\fnm\right\|_{\lpt}\left\|\Aonelz \left(P_n^*P_n-I\right)g\right\|_{\lpt}
	\to0,
	$$
since $\left\|P_n^*\fnm\right\|_{\lpt}$ is bounded, and since $\Aonelnm-\Aonelz$ tends strongly to $0$ as an operator $\hpt\to\lpt$ and $P_n^*P_n$ tends strongly to $I$ in $\hpt$ by Lemma \ref{conv-h2}. Finally
	$$
	III
	=
	\left|\left< P_n^*\fnm-\phiz,\Aonelz g\right>_{\lpt}\right|\\
	\to
	0,
	$$
since $P_n^*\fnm\to\phiz$ strongly in $\lpt$ by assumption.

\end{proof}

\begin{lem}\label{neg-eig-mz}
For $n$ sufficiently large, $\Mnz$ has exactly $K_n:=n-neg(\Aonez)+neg(\Atwoz)+neg(l^0)$ negative eigenvalues.
\end{lem}

\begin{proof}
Since $\Mnz$ is diagonal, we may consider each entry along the diagonal separately. The number of negative eigenvalues of $-P_n\Aonez P^*_n$ equals the number of positive eigenvalues of $P_n\Aonez P^*_n$, namely $pos(P_n\Aonez P^*_n)$. Since we assume that $\ker{\Aonez}=\{\text{constants}\}$, and since our domain does not include constant functions, $P_n\Aonez P^*_n$ has no null space. Therefore, $pos(P_n\Aonez P^*_n)=n-neg(P_n\Aonez P^*_n)$. However, for $n$ that is sufficiently large $neg(P_n\Aonez P^*_n)=neg(\Aonez)$ by our definition of $P_n$. Thus, the contribution from the first term on the diagonal is $n-neg(\Aonez)$. The contribution to the negative spectrum from the next term on the diagonal is $neg(\Atwoz)=neg(Q_n\Atwoz Q^*_n)$ for $n$ that is sufficiently large. The contribution from the last term depends upon the sign of $l^0$ and is denoted by $neg(l^0)$. Combining these observations we get our result.
\end{proof}

\begin{prop}\label{lemma2}
There exist $N\in\N$ and $\ula>0$ such that for any $n>N$ and for all $\la\in\left[0,\ula\right]$, $neg(\Mnl)\geq neg(\Mnz)=K_n$.
\end{prop}



\begin{proof}
%
We show the existence of $\ula$ and $N$ by contradiction: If they do not exist, then for any $k>1$ there exist $n_k>k$ and $\la_k'<\frac{1}{k}$
for which $neg(\mathcal{M}_{n_k}^{\la_k'})<K_{n_k}$. We choose $k$ large enough for Lemma \ref{neg-eig-mz} to hold, and fix it. Since $\mathcal{M}_{n_k}^\la$ is a continuous mapping in $\la$, its spectrum varies continuously with $\la$. Therefore, since $neg(\mathcal{M}_{n_k}^0)=K_{n_k}>neg(\mathcal{M}_{n_k}^{\la_k'})$, at least one eigenvalue in the negative part of $\Sigma(\Mnl)$ must cross $0$ from left to right as $\la$ varies from $0$ to $\la_k'$.
Since there is no spectrum of $\mathcal{M}_{n_k}^0$ on the interval $(\sigma^*,0)$, this eigenvalue must also cross this interval, and, in particular, it must cross $\sigma^*$ at some value $\la_k\in(0,\la_k')$.


To summarize this argument, our contradiction asserts that there exist $n_k\to\infty$ and $\la_k\to0$ for which
%
%
	\begin{equation}\label{kernel}
	\mathcal{M}_{n_k}^{\la_k} u_{n_k}
	=
	\sigma^* u_{n_k},
	\end{equation}
where $0\neq u_{n_k}^T=(\phi_{n_k},\psi_{n_k},b_{n_k})\in\R^n\times\R^n\times\R$. To simplify notation we drop the ``$k$'' index and simply write $\la_n$ instead of $\la_k$ and $n$ instead of $n_k$. Our plan is to show that this contradicts the fact that $\sigma^*$ is not an eigenvalue of $\Mz$, as follows: First we show that $u_n$ has some nontrivial limit in an appropriate space; then we use Lemma \ref{eq-converge} to show that the operators $\Mklk$ converge (in the weak sense defined in the statement of that lemma) to $\Mz$.

Since $P_n^*$ and $Q_n^*$ both have trivial kernels, we may normalize the eigenvectors as follows
	\begin{equation}\label{one}
	\|P_n^*\fk\|_{\lpt}+
	\|Q_n^*\psi_n\|_{\lpt}+
	|b_n|=1,
	\end{equation}
and, therefore, there exist the two weak limits in $\lpt$ and the limit in $\R$, respectively (after extracting a subsequence)
	\begin{equation}\label{weak-limits}
	P_n^*\fnm\rightharpoonup\phi,
	\hspace{1cm}
	Q_n^*\pnm\rightharpoonup\psi,
	\hspace{1cm}
	b_n\to b.
	\end{equation}
Our goal is to show that $(\phi,\psi,b)$ is nontrivial, and that
	\begin{equation}\label{contradiction}
	\Mz
	\left(\begin{array}{c}
	\phi\\
	\psi\\
	b
	\end{array}\right)=
	\sigma^*
	\left(\begin{array}{c}
	\phi\\
	\psi\\
	b
	\end{array}\right),
	\end{equation}
thus reaching a contradiction to the fact that $\sigma^*\notin\ker(\Mz)$. Note that we do not have to worry about showing that $(\phi,\psi,b)$ is not a multiple of $(1,0,0)$, since $(1,0,0)$ is in the kernel of $\Ml$ for all $\la\geq0$, but $\sigma^*\neq0$. We begin by showing that both $P_n^*\fnm$ and $P_n^*\pnm$ are bounded in $H_P^1$. The first row of \eqref{kernel} is
	\begin{equation}\label{line1}
	P_n\left(-\Aonelnm P_n^*\fnm+\Blnm Q_n^*\pnm+\Clnm\bnm\right)=\sigs\fnm.
	\end{equation}
Write $\phi_n=\left(\phi_n^1,\phi_n^2,\dots,\phi_n^n\right)$, and take the inner product in $\R^n$ of \eqref{line1} with $\fnm$:
	\begin{eqnarray*}
	\sigs|\fnm|^2&=&\fnm\cdot P_n\left(-\Aonelnm P_n^*\fnm+\Blnm Q_n^*\pnm+\Clnm\bnm\right)\\
	&=&
	\left<P_n^*\fnm,-\Aonelnm P_n^*\fnm\right>_{\lpt}
	+\fnm\cdot P_n\left(\Blnm Q_n^*\pnm+\Clnm\bnm\right)
	=I+II.
	\end{eqnarray*}
Since $1\geq\|P_n^*\phi_n\|^2_{\lpt}=\sum_{k=1}^n|\phi_n^k|^2=|\phi_n|^2$, the left hand side of the above equation is uniformly (in $n$) bounded. Moreover,
	\begin{eqnarray*}
	I
	&=&
	\left<P_n^*\fnm,-\Aonelnm P_n^*\fnm\right>_{\lpt}\\
	&=&
	\left<P_n^*\fnm,\p_x^2 P_n^*\fnm\right>_{\lpt}+\left<P_n^*\fnm,\left(\sum_\pm\int\mu^\pm_e\;dv\right)P_n^*\fnm-\sum_\pm\int\mu^\pm_e\Qlnm_\pm \left(P_n^*\fnm\right)dv\right>_{\lpt}\\
	&=&
	-\left\|\p_x P_n^*\fnm\right\|_{\lpt}^2+I_1+I_2
	\end{eqnarray*}
where it is easily seen that
	$$
	|I_1|
	=
	\left|\left<P_n^*\fnm,\left(\sum_\pm\int\mu^\pm_e\;dv\right)P_n^*\fnm\right>_{\lpt}\right|
	\leq
	\left\|P_n^*\fnm\right\|^2_{\lpt}\sum_\pm\sup_{x}\int|\mu^\pm_e|\;dv
	<
	C<\infty
	$$
and, as in the proof of Lemma \ref{propql}\eqref{qlnorm}, we have
	\begin{align*}
	|I_2|
	&=
	\left|\int_0^PP_n^*\phi_n\sum_\pm\int\mu^\pm_e\Qlnm_\pm(P_n^*\phi_n)dv\;dx\right|\\
	&=
	\left|\sum_\pm\int_{-\infty}^0\la_n e^{\la_n s}\int_0^P\int\mu^\pm_e\;P_n^*\phi_n(x)\;\;P_n^*\phi_n(X^\pm(s))\;dv\;dx\;ds\right|\\
	&\leq
	2\left\|P_n^*\phi_n\right\|_w^2
	\leq
	C\left\|P_n^*\phi_n\right\|_{\lpt}^2
	<C<
	\infty.
	\end{align*}
In addition,
	\begin{eqnarray*}
	|II|
	&=&
	\left|\fnm\cdot P_n\left(\Blnm Q_n^*\pnm+\Clnm\bnm\right)\right|\\
	&=&
	\left|\left<P_n^*\fnm,\Blnm Q_n^*\pnm+\Clnm\bnm\right>_{\lpt}\right|\\
	&\leq&
	\left\|P_n^*\fnm\right\|_{\lpt}\left(\left\|\Blnm\right\|_{\lpt\to\lpt}\left\| Q_n^*\pnm\right\|_{\lpt}
	+
	\left\|\Clnm\bnm\right\|_{\lpt}\right)\\
	&<&
	C<\infty.
	\end{eqnarray*}
by Lemma \ref{properties2}.
Thus $\left\|\p_x P_n^*\fnm\right\|_{\lpt}$ is bounded \emph{uniformly} in $n$ in $\lpt$, so that $\left\|P_n^*\fnm\right\|_{H_P^1}$ is uniformly bounded. Therefore $P_n^*\fnm$ converges strongly in $\lpt$ and weakly in $H^1_P$ to $\phi$ by \eqref{weak-limits}. Similarly, we show that $Q_m^*\pnm$ is bounded in $H_P^1$ by considering the second row of \eqref{kernel}:
\begin{equation}
Q_m\left(\left(\Blnm\right)^* P_n^*\fnm+\Atwolnm P_n^*\pnm-\Dlnm\bnm\right)=\sigs\pnm.
\end{equation}
The analysis is similar. Hence $P_n^*\fnm$ and $P_n^*\pnm$ have strong limits in $\lpt$ and weak limits in $H^1_P$, which must be $\phi$ and $\psi$ respectively. Due to \eqref{one} the limit $\left(P_n^*\fnm,P_n^*\pnm,b_n\right)$ cannot be trivial: $\left(\phi,\psi,b\right)\neq(0,0,0)$.

It remains to be shown that, in fact, $\phi$ and $\psi$ lie in $\hpt$ and hence, in the domain of $\Mz$. Consider again equation \ref{line1}: We rewrite this equation, keeping only the Laplacian on the left-hand-side, moving all other terms to the right-hand-side and applying $P_n^*$. We get
	\begin{align*}
	P_n^*P_n\p_x^2\left(P_n^*\fnm\right)
	=&
	-P_n^*P_n\left(\sum_\pm\int\mu_e^\pm\;dv\right)P_n^*\fnm
	+
	P_n^*P_n\sum_\pm\int\mu_e^\pm\Qlnm_\pm\left(P_n^*\fnm\right)dv\\
	&
	-P_n^*P_n\Blnm Q_n^*\pnm-P_n^*P_n\Clnm\bnm+\sigma^*P_n^*\phi_n.
	\end{align*}
The last term may be written as $\sigma^*P_n^*P_nP_n^*\phi_n$, so that we may denote the entire right-hand-side as $P_n^*P_nh_n$ for brevity. Now, we know that $h_n\in\lpt$ has a strong limit in $\lpt$, and $P_n^*P_n\to I$ strongly in $\lpt$. Therefore, the right-hand-side converges in $\lpt$. Using Lemma \ref{h-1} and since $P_n^*\phi_n$ converges weakly in $H^1_P$ to $\phi$, the left hand side converges weakly in the $H^{-1}_P$ sense to $\p_x^2\phi$. By elliptic regularity one can bootstrap and deduce that, in fact, $\phi\in\hpt$.
	
Finally, using Lemma \ref{eq-converge}, we know that the approximate equations \eqref{kernel} tend (weakly, in the sense of the lemma) to the equations $\Mz u=\sigs u$. But $\sigma^*\notin\Sigma(\Mz)$. This contradiction ends the proof.
\end{proof}

\begin{prop}\label{trivial-atwoz}
If the null space of $\Atwoz$ is trivial, then for any $n>N$ and for all $\la\in\left[0,\ula\right]$, $neg(\Mnl)= neg(\Mnz)=K_n$, where $N$ and $\ula$ are as in Proposition \ref{lemma2}.
\end{prop}

\begin{proof}
Using the same contradiction argument as in Proposition \ref{lemma2}, we can show that $pos(\Mnl)\geq pos(\Mnz)$ for small $\la$ and large $n$. We conclude that an increase in either the number of negative eigenvalues or positive eigenvalues can only be due to zero eigenvalues of $\Mnz$ that move left or right as $\la$ increases. However, under the assumption that the null space of $\Atwoz$ is trivial, $\Mnz$ has a trivial kernel as well, and, therefore both $neg(\Mnl)$ and $pos(\Mnl)$ must remain constant (and equal to their values when $\la=0$) for small $\la$ and large $n$.
\end{proof}
\begin{lem}\label{large}
There exists $\La^*\geq\La$ ($\La$ as in Lemma \ref{properties1}) such that for every $n\in\N$ and for any $\la\geq\La^*$, $\Mnl$ has precisely $n+1$ negative eigenvalues.
\end{lem}

\begin{proof}
Since $\Mnl:\R^n\times\R^n\times\R\to\R^n\times\R^n\times\R$ is symmetric, it has $2n+1$ eigenvalues, all real. Letting $\psi\in\R^n$, we have

	\begin{equation}
	\left<\Mnl \left(\begin{array}{c}0\\\psi\\0\end{array}\right),\left(\begin{array}{c}0\\\psi\\0\end{array}\right)\right>_{\R^{2n+1}}
	=
	\left<Q_n\Atwol Q_n^*\psi,\psi\right>_{\R^{n}}
	=
	\left<\Atwol \;Q_n^*\psi,Q_n^*\psi\right>_{\lpt}
	>
	0
	\end{equation}
for all $\la>\La$ by Lemma \ref{properties1}. This implies that $\Mnl$ is positive definite on a subspace of dimension $n$, and, therefore it has at least $n$ positive eigenvalues. Similarly, we now show that there exists a subspace of dimension $n+1$ on which $\Mnl$ is negative definite: Let $(\phi,0,b)\in\R^n\times\R^n\times\R$ and consider
	\begin{align}\label{neg-def}
	\left<\Mnl \left(\begin{array}{c}\phi\\0\\b\end{array}\right),\left(\begin{array}{c}\phi\\0\\b\end{array}\right)\right>_{\R^{2n+1}}
	&=
	-\left<\Aonel P_n^*\phi,P_n^*\phi\right>_{\lpt}
	+
	2\left<\Cl b,P_n^*\phi\right>_{\lpt}
	-
	P(\la^2-l^\la)b^2.
	\end{align}
We estimate the middle term as follows:
	\begin{align*}
	2\left|\left<\Cl b,P_n^*\phi\right>_{\lpt}\right|
	\leq
	2\left\|\Cl b\right\|_{\lpt}\left\|P_n^*\phi\right\|_{\lpt}
	\leq
	\frac{\left\|\Cl b\right\|^2_{\lpt}}{\epsilon^2}+\epsilon^2\left\|P_n^*\phi\right\|^2_{\lpt}.
	\end{align*}
Letting $\epsilon^2=\frac{1}{\la}$, we have
	\begin{align*}
	\left<\Mnl \left(\begin{array}{c}\phi\\0\\b\end{array}\right),\left(\begin{array}{c}\phi\\0\\b\end{array}\right)\right>_{\R^{2n+1}}
	\leq
	-\left<\Aonel P_n^*\phi,P_n^*\phi\right>_{\lpt}
	+
	\frac{\left\|P_n^*\phi\right\|^2_{\lpt}}{\la}
	-
	P(\la^2-l^\la)b^2
	+
	\la\left\|\Cl b\right\|^2_{\lpt}.
	\end{align*}
Using the fact that $\Aonel>\gamma>0$, for all $\la>\La$ (see Lemma \ref{properties1}\eqref{operators-positive}), this expression is negative for all $\phi\in\R^n$ and $b\in\R$, since $l^\la$ and $\Cl$ are both uniformly bounded. Therefore, there exists  a $\La^*>0$ such that for every $\la\geq\La^*$ there exists an $n+1$ dimensional subspace on which $\Mnl$ is negative definite. We conclude that
	\begin{equation}
	neg\left(\Mnl\right)=n+1,
	\hspace{.5cm}
	\text{for all }\la>\La^*.
	\end{equation}
Notice that $\La^*$ does not depend upon $n$.
\end{proof}

\section{Limit as $n\to\infty$}\label{sec5}
\begin{lem}\label{solution}
Let $\la^*,\La^*, N$ be as above.
Fix any $n>N$. Then there exists $\la_{n}\in[\la^*,\La^*]$ such that $\Mnmlnm$ has a nontrivial kernel.
\end{lem}

\begin{proof}
As we have seen above, $\la^*$ and $\La^*$ do not depend on $n$. We apply a simple continuity argument: $\Mnl$ is continuous in $\la$ for each (fixed) $n$ in the sense that if $\sigma>0$, then there exist $C,\delta>0$ such that
$$\|\Mnl-\Mns\|\leq C|\la-\sigma|$$
for $\la\in(0,\infty)$ and $|\la-\sigma|<\delta$. This follows from Lemma \ref{propml}.
By Proposition \ref{lemma2}, $\mathcal{M}^{\la^*}_{n}$ has at least $n-neg\left(\Aonez\right)+neg\left(\Atwoz\right)+neg\left(l^0\right)$ negative eigenvalues. By Lemma \ref{large}, $\mathcal{M}^{\La^*}_{n}$ has exactly $n+1$ negative eigenvalues. Since $\Mnl$ is a finite-dimensional operator, its set of eigenvalues varies continuously with $\la$. Thus, if

	\begin{equation}\label{criterion}
	n-neg\left(\Aonez\right)+neg\left(\Atwoz\right)+neg\left(l^0\right)
	>
	n+1
	\end{equation}
then at least one eigenvalue must cross $0$ for some $\la_{n}\in(\la^*,\La^*)$. In particular, the $\la$ value for which 0 is a (nontrivial) eigenvalue, has a corresponding (nontrivial) eigenspace. Unraveling condition \eqref{criterion}, we get the equivalent criterion:
	\begin{equation}
	neg\left(\Atwoz\right)>neg\left(\Aonez\right)+neg(-l^0)
	\end{equation}
which is precisely the main assumption of Theorem \ref{mainthm}. In the context of Theorem \ref{mainthm2}  one would invoke Proposition \ref{trivial-atwoz} instead of invoking Proposition \ref{lemma2} and get the criterion
	\begin{equation}
	neg\left(\Atwoz\right)\neq neg\left(\Aonez\right)+neg(-l^0).
	\end{equation}
\end{proof}

\begin{lem}\label{exist}
There exists $0<\la_0<\infty$ and a nontrivial $u_0^T=\left(\phi_0, \psi_0, b_0\right)$ that is not a multiple of $\left(1,0, 0\right)$ such that
\begin{equation}\label{main}
\Mlz u_0=0.
\end{equation}
\end{lem}

\begin{proof}
This proof is very similar to the proof of Proposition \ref{lemma2}: We have seen that for all $n>N$, $\Mnl$ has a nontrivial kernel when $\la=\la_{n}$ if \eqref{criterion} is satisfied. We show that \eqref{main} is satisfied weakly, with $u_0$ extracted by some compactness argument from nontrivial elements in the kernel of $\Mnmlnm$, denoted by $u_{n}$.
As we have seen in Lemma \ref{solution}, the equation

	\begin{equation}\label{ker1}
	\Mnmlnm u_{n}
	=
	\matthree{-P_n\Aonelnm P_n^*}{P_n\Blnm Q_n^*}{P_n\Clnm}
	{Q_n\left(\Blnm\right)^*P_n^*}{Q_n \Atwolnm Q_n^*}{-Q_n\Dlnm}
	{\left(\Clnm\right)^*P_n^*}{-\left(\Dlnm\right)^*Q_n^*}{-P\left(\la_{n}^2-l^{\la_{n}}\right)}
	\left(\begin{array}{c}
	\fnm\\
	\pnm\\
	\bnm
	\end{array}\right)
	=
	\left(\begin{array}{c}
	0\\
	0\\
	0
	\end{array}\right)
	\end{equation}
has a nontrivial solution for all $n>N$. Here $0<\la^*<\la_{n}<\La^*<\infty$. Let us extract a subsequence $\la_{n}\to\la_0$. We want to show that the ``limiting" equation

	\begin{equation}\label{ker2}
	\Mlz u_{0}
	=
	\matthree{-\Aonelz}{\Blz}{\Clz}
	{\left(\Blz\right)^*}{\Atwolz}{-\Dlz}
	{\left(\Clz\right)^*}{-\left(\Dlz\right)^*}{-P\left(\la_{0}^2-l^{\la_{0}}\right)}
	\left(\begin{array}{c}
	\phi_0\\
	\psi_0\\
	b_0
	\end{array}\right)
	=
	\left(\begin{array}{c}
	0\\
	0\\
	0
	\end{array}\right)
	\end{equation}
is satisfied nontrivially.
%
%
We follow the procedure of Proposition \ref{lemma2}, showing that $\left(P_n^*\fnm,Q_n^*\pnm,\bnm\right)=u_n^T\to u_0^T=(\phi_0,\psi_0,b_0)\neq(0,0,0)$ in $\hpt\times\hpt\times\R$:

\begin{enumerate}
\item
We normalize the vectors $u_n$ as in \eqref{one}:
	\begin{equation}\label{one2}
	\|P_n^*\fk\|_{\lpt}+
	\|Q_n^*\psi_n\|_{\lpt}+
	|b_n|=1.
	\end{equation}

\item
We take the inner product of the first row of \eqref{ker1} with $\phi_n$, to obtain the equation
	\begin{equation}\label{line12}
	\phi_n\cdot P_n\left(-\Aonelnm P_n^*\fnm+\Blnm Q_n^*\pnm+\Clnm\bnm\right)=0.
	\end{equation}
Showing uniform boundedness of $P_n^*\phi_n$ in $H_P^1$ is identical to the calculations performed in the lines following \eqref{line1}: We show that the $\lpt$ norm of all terms in \eqref{line12} is uniformly bounded in $n$ (except for the Laplacian), and, therefore, by integrating by parts we obtain the uniform $H_P^1$ bound. We conclude that $P_n^*\phi_n$ converges in $\lpt$. 

\item
Using the second row of \eqref{ker1} we show that $Q_n^*\psi_n$ converges in $\lpt$.

\item
We bootstrap our convergence problem, showing that, in fact, $P_n^*\phi_n$ converges to $\phi$ in $\hpt$, by applying $P_n^*$ to the first row of \eqref{ker1}. Similarly, to show that $Q_n^*\psi_n\to\psi$ in $\hpt$ we apply $Q_n^*$ to the second row of \eqref{ker1}.

\end{enumerate}
Moreover, since $P_n^*\phi_n\in\hptz$ all have mean $0$, so does $\phi_0$. Therefore $u_0$ is not a multiple of $(1,0,0)$.
The fact that the terms in the equations \eqref{ker1} tend to the terms in the equations \eqref{ker2} weakly follows from Lemma \ref{eq-converge}. This finishes the proof.

\end{proof}

\section{Construction of a Growing Mode}\label{construction}
We finish the proof of Theorem \ref{mainthm} by verifying that the nontrivial element $u_0$ that we found above satisfies the linearized RVM System. For ease of notation, we drop the ``0" subscript, so that we simply have $u^T=\left(\phi, \psi, b\right)$, and $\la$. The equation we verified in the previous section is
	\begin{equation}\label{final}
	\Ml u
	=
	\matthree{-\Aonel}{\Bl}{\Cl}
	{\left(\Bl\right)^*}{\Atwol}{-\Dl}
	{\left(\Cl\right)^*}{-\left(\Dl\right)^*}{-P\left(\la^2-l^\la\right)}
	\left(\begin{array}{c}\phi\\\psi\\b\end{array}\right)
	=
	0,
	\end{equation}
with $u$ nontrivial and not a multiple of $(1,0,0)$, and where $0<\la<\infty$.
We begin by defining $f^\pm(x,v)$:
	\begin{equation}
	f^\pm(x,v)
	=
	\pm\mu^\pm_e\phi(x)\pm\mu^\pm_p\psi(x)
	\mp\mu^\pm_e\left[\Ql_\pm\phi-\Ql_\pm\left(\hat{v}_2\psi\right)-b\Ql_\pm\hat{v}_1\right].
	\end{equation}

In addition, we define
	$$
	E_1
	=
	-\p_x\phi-\la b
	\hspace{1.8cm}
	E_2
	=
	-\la\psi
	\hspace{1.8cm}
	B
	=
	\p_x\psi
	$$
and
	$$	
	\rho
	=
	\int (f^+-f^-)\;dv
	\hspace{1.5cm}
	j_i
	=
	\int \hat{v}_i(f^+-f^-)\;dv,
	\hspace{.3cm}
	i=1,2.
	$$
	
\begin{lem}
Gauss' Equation \eqref{gauss1d} holds. 
\end{lem}

\begin{proof}
We use the first row of \eqref{final}.
	\begin{eqnarray*}
	\p_xE_1
	&=&
	-\p_x^2\phi\\
	&=&
	\sum_\pm\left\{\int\mu^\pm_e\;dv\;\phi
	-\int\mu^\pm_e\Ql_\pm\phi\;dv
	+\int\mu^\pm_p\;dv\;\psi
	+\int\mu^\pm_e\Ql_\pm(\hat{v}_2\psi)\;dv
	+\int\mu^\pm_e\Ql_\pm\left(\hat{v}_1\right)\;dv\; b\right\}\\
	&=&
	\sum_\pm\int\left(\mu^\pm_e\phi
	+\mu^\pm_p \psi
	-\mu^\pm_e\left[\Ql_\pm\phi-\Ql_\pm\left(\hat{v}_2\psi\right)-b\Ql_\pm\hat{v}_1\right]\right)dv\\
	&=&
	\int(f^+-f^-)\;dv=\rho.
	\end{eqnarray*}
\end{proof}
\begin{lem}
The Linearized Vlasov Equation \eqref{linvlasov} holds.
\end{lem}

\begin{proof}
We recall that the Linearized Vlasov Equation \eqref{linvlasov} is
	$$
	\left(\p_t+D^\pm\right)f^\pm
	=
	\mp\mu^\pm_e\hat{v}_1E_1\pm\mu^\pm_p\hat{v}_1B\mp\left(\mu^\pm_e\hat{v}_2+\mu^\pm_p\right)E_2.
	$$

We let $g\in C_c^1([0,P]\times\R^2)$ be any test function. We show it for the electrons, $f^-$, and drop all ``$-$" superscripts. Showing that the linearized Vlasov Equation holds for $f^+$ is identical. We write:

	\begin{eqnarray*}
	\int_0^P\int	(Dg)f		\;dv\;dx
	&=&
	\int_0^P\int	(Dg)\left(-\mu_e\phi(x)-\mu_p\psi(x)
	+\mu_e\left[\Ql\phi-\Ql\left(\hat{v}_2\psi\right)-b\Ql\hat{v}_1\right]\right)\;dv\;dx\\
	&=&
	I+II+III+IV+V.
	\end{eqnarray*}
	
For the terms $I$ and $II$ we use the fact that $D$ is skew-adjoint and that $\mu$ is invariant under $D$, to ``move" the operator $D$ over from $g$ to $\phi$ and $\psi$ respectively. Thus we focus on the terms $III, IV, V$, where the definition of $\Ql$ is important, and the fact that $(x,v)\to(X,V)$ has Jacobian = 1.

	\begin{eqnarray*}
	III
	&=&
	\int_0^P\int	(Dg)\mu_e\Ql\phi	\;dv\;dx	\\
	&=&
	\int_{-\infty}^0	\la e^{\la s}	\int_0^P\int	\mu_e (Dg)(x,v)\phi(X(s;x,v))\;dv\;dx\;ds\\
	&=&
	\int_{-\infty}^0	\la e^{\la s}	\int_0^P\int	\mu_e (Dg)(X(-s),V(-s))\phi(x)\;dv\;dx\;ds\\
	&=&
	\int_0^P\int	\mu_e	\int_{-\infty}^0	\la e^{\la s}\left(-\frac{d}{ds}g(X(-s),V(-s))\right)ds\;
	\phi(x)\;dv\;dx\\
	&=&
	\int_0^P\int\mu_e\left\{-\la g(x,v)
	+\int_{-\infty}^0\la^2 e^{\la s}g(X(-s),V(-s))\;ds\;\right\}
	\phi(x)\;dv\;dx\\
	&=&
	\int_0^P\int\left\{-\mu_e\la \phi(x)
	+\mu_e\int_{-\infty}^0\la^2 e^{\la s}\phi(X(s),V(s))\;ds\;\right\}
	g(x,v)\;dv\;dx\\
	&=&
	\la\int_0^P\int\left\{-\mu_e\phi+\mu_e\Ql\phi\right\}g\;dv\;dx.
	\end{eqnarray*}
	
Similarly

	$$
	IV
	=
	-\la\int_0^P\int\left\{-\mu_e\hat{v}_2\psi+\mu_e\Ql(\hat{v}_2\psi)\right\}g\;dv\;dx
	$$
and
	$$
	V
	=
	-\la\int_0^P\int\left\{-b\mu_e\hat{v}_1+b\mu_e\Ql\hat{v}_1\right\}g\;dv\;dx.
	$$
	
Thus

	\begin{eqnarray*}
	\int_0^P\int	(Dg)f		\;dv\;dx
	&=&
	\int_0^P\int\left\{
	\mu_eD\phi+\mu_pD\psi
	\right\}g\;dv\;dx			\\
	&&+\la\int_0^P\int\left\{
	-\mu_e\phi+\mu_e\Ql\phi
	+\mu_e\hat{v}_2\psi-\mu_e\Ql(\hat{v}_2\psi)
	+b\mu_e\hat{v}_1-b\mu_e\Ql\hat{v}_1
	\right\}g\;dv\;dx			\\
	&=&
	\int_0^P\int\la\left\{
	-\mu_e\phi-\mu_p\psi+\mu_e\left[\Ql\phi-\Ql\left(\hat{v}_2\psi\right)-b\Ql\hat{v}_1\right]
	\right\}g\;dv\;dx			\\
	&&+\int_0^P\int\left\{
	\la\mu_p\psi+\mu_eD\phi+\mu_pD\psi+\la\mu_e\hat{v}_2\psi+\la b\mu_e\hat{v}_1
	\right\}g\;dv\;dx			\\
	&=&
	\int_0^P\int\left\{
	\la\left(f+\mu_p\psi\right)
	+\mu_eD\phi+\mu_pD\psi+\la\mu_e\hat{v}_2\psi+\la b\mu_e\hat{v}_1
	\right\}g\;dv\;dx			\\
	\end{eqnarray*}
	
Therefore, weakly, $f$ satisfies the equation

	\begin{eqnarray*}
	(\la+D)f
	&=&
	-\mu_eD\phi-\mu_pD\psi-\la\mu_p\psi-\la\mu_e\hat{v}_2\psi-\la b\mu_e\hat{v}_1\\
	&=&
	-\mu_e\hat{v}_1\p_x\phi-\mu_p\hat{v}_1\p_x\psi+\mu_pE_2-\la\mu_e\hat{v}_2\psi-\la b\mu_e\hat{v}_1\\
	&=&
	\mu_e\hat{v}_1E_1-\mu_p\hat{v}_1B+\left(\mu_p+\mu_e\hat{v}_2\right)E_2,
	\end{eqnarray*}
	
which is precisely \eqref{linvlasov}.
\end{proof}
\begin{lem}[Continuity Equation]
The relation
	$
	\p_xj_1+\la\rho=0
	$
holds.
\end{lem}

\begin{proof}
Integrating the Linearized Vlasov Equation with respect to $v$, we get
	$$
	\int\left(\p_t+D^\pm\right)f^\pm\;dv
	=
	\int\left(\mp\mu^\pm_e\hat{v}_1E_1\pm\mu^\pm_p\hat{v}_1B\mp\left(\mu^\pm_e\hat{v}_2+\mu^\pm_p\right)E_2\right)dv
	=
	0,
	$$
where we use the facts that $\mu^\pm$ is even in $v_1$, and that $\p\mu^\pm/\p v_2=\mu^\pm_e\hat{v}_2+\mu^\pm_p$ is a perfect derivative. Subtracting these two equations, replacing the time derivative by a factor of $\la$, and using the fact that $D^\pm$ consist of three terms, of which only the first is not a $v_i$ derivative, we have:

	$$
	0
	=
	\sum_\pm\int\left(\la+D^\pm\right)f^\pm\;dv
	=
	\la\rho+\p_x\sum_\pm\int\hat{v}_1f^\pm\;dv
	=
	\la\rho+\p_xj_1.
	$$
\end{proof}
\begin{lem}
Amp\`{e}re's Equations \eqref{ampere1d1} and \eqref{ampere1d2}
hold.
\end{lem}

\begin{proof}
We first want to show that

	\begin{equation}\label{amp}
	\la E_1
	=
	-j_1.
	\end{equation}

Recalling the definition $E_1=-\p_x\phi-\la b$, we wish to show that
	\begin{equation}\label{amp2}
	\la^2b=-\la\p_x\phi+j_1.
	\end{equation}

Let us show equality of the derivative with respect to $x$, and then equality of the integral with respect to $x$. Differentiating once, and using the continuity equation, we get

	\begin{eqnarray*}
	0
	&=&
	-\p_x^2\phi+\frac{1}{\la}\p_xj_1\\
	&=&
	-\p_x^2\phi-\rho\\
	&=&
	-\p_x^2\phi
	-\sum_\pm\int\left(\mu^\pm_e\phi
	+\mu^\pm_p \psi
	-\mu^\pm_e\left[\Ql_\pm\phi-\Ql_\pm\left(\hat{v}_2\psi\right)-b\Ql_\pm\hat{v}_1\right]\right)dv\\
	&=&
	\Aonel\phi-\Bl\psi-\Cl b
	\end{eqnarray*}
which is precisely the first row of \eqref{final}. This verifies that the derivatives are the same.

Now we turn to the integral of \eqref{amp}: Plugging in the the relation $\la^2b=\frac{1}{P}\left(\Cl\right)^*\phi-\frac{1}{P}\left(\Dl\right)^*\psi+l^\la b$ (which is obtained from the last row of \eqref{final}) into \eqref{amp2}, it suffices to show
	$$
	-\la\p_x\phi+j_1
	=
	\frac{1}{P}\left(\Cl\right)^*\phi-\frac{1}{P}\left(\Dl\right)^*\psi+l^\la b.
	$$
Writing in detail the expressions for $j_1$ and for the operators on the right hand side, we need to show that
	\begin{eqnarray*}
	&&-\la\p_x\phi+\sum_\pm\int\hat{v}_1
	\left(\mu^\pm_e\phi
	+\mu^\pm_p \psi
	-\mu^\pm_e\left[\Ql_\pm\phi-\Ql_\pm\left(\hat{v}_2\psi\right)-b\Ql_\pm\hat{v}_1\right]\right)dv\\
	&&=
	\frac{1}{P}\sum_\pm\int_0^P\int\left[\mu^\pm_e\Ql_\pm\left(\hat{v}_1\right)\phi
	-\hat{v}_2\mu^\pm_e\Ql_\pm\left(\hat{v}_1\right)\psi
	+b\hat{v}_1\mu^\pm_e\Ql_\pm\left(\hat{v}_1\right)\right]\;dv\;dx.
	\end{eqnarray*}
Since $\mu$ is even in $v_1$, we may drop the first two terms in the integral on the left hand side. Therefore we need to show that
	\begin{eqnarray*}
	-\la\p_x\phi
	&=&
	\sum_\pm\int\hat{v}_1\mu^\pm_e
	\left[\Ql_\pm\phi-\Ql_\pm\left(\hat{v}_2\psi\right)-b\Ql_\pm\left(\hat{v}_1\right)\right]dv\\
	&&+
	\frac{1}{P}\sum_\pm\int_0^P\int\mu^\pm_e\left[\Ql_\pm\left(\hat{v}_1\right)\phi
	-\hat{v}_2\Ql_\pm\left(\hat{v}_1\right)\psi
	+b\hat{v}_1\Ql_\pm\left(\hat{v}_1\right)\right]dv\;dx.\\
	\end{eqnarray*}
Integrating this equation over the period $P$ we get
	\begin{eqnarray*}
	0
	&=&
	\sum_\pm\int_0^P\int\hat{v}_1\mu^\pm_e
	\left[\Ql_\pm\phi-\Ql_\pm\left(\hat{v}_2\psi\right)-b\Ql_\pm\left(\hat{v}_1\right)\right]dv\;dx\\
	&&+
	\sum_\pm\int_0^P\int\mu^\pm_e\left[\Ql_\pm\left(\hat{v}_1\right)\phi
	-\hat{v}_2\Ql_\pm\left(\hat{v}_1\right)\psi
	+b\hat{v}_1\Ql_\pm\left(\hat{v}_1\right)\right]dv\;dx\\
	&=&
	\sum_\pm\int_0^P\int\hat{v}_1\mu^\pm_e
	\left[\Ql_\pm\phi-\Ql_\pm\left(\hat{v}_2\psi\right)\right]dv\;dx
	+
	\sum_\pm\int_0^P\int\mu^\pm_e\left[\Ql_\pm\left(\hat{v}_1\right)\phi
	-\hat{v}_2\Ql_\pm\left(\hat{v}_1\right)\psi\right]dv\;dx\\
	\end{eqnarray*}
which indeed holds due to the change of variables $(x,v)\to(X,V)$. Therefore, both the derivatives and the integrals are equal. Hence \eqref{ampere1d1} holds.

Now we turn to show that the equation $\la E_2+\p_xB=-j_2$ holds. We again recall our definitions
	
	$$
	E_2=-\la\psi,
	\hspace{.5cm}
	B=\p_x\psi,
	\hspace{.5cm}
	j_2=\int\hat{v}_2(f^+-f^-)\;dv,
	$$
that imply that we need to show
	\begin{eqnarray*}
	-\la^2\psi+\p_x^2\psi
	&=&
	-\int\hat{v}_2(f^+-f^-)\;dv\\
	&=&
	-\sum_\pm\int\hat{v}_2\left(\mu^\pm_e\phi
	+\mu^\pm_p \psi
	-\mu^\pm_e\left[\Ql_\pm\phi-\Ql_\pm\left(\hat{v}_2\psi\right)-b\Ql_\pm\hat{v}_1\right]\right)dv.
	\end{eqnarray*}
But this is precisely the second row of \eqref{final}.


\end{proof}
%
%
Our proof of Theorem \ref{mainthm} is now complete.

\begin{proof}[Proof of Theorem \ref{mainthm2}]
The proof of Theorem \ref{mainthm2} is identical to the proof of Theorem \ref{mainthm}, with the only difference being that in Lemma \ref{solution} we invoke Proposition \ref{trivial-atwoz} instead of invoking Proposition \ref{lemma2}.
\end{proof}

\section{Examples}\label{sec-examples}


\subsection{Homogeneous Example}
We start with the simple homogeneous case, where there is no $x$ dependence. This case is so simple, that certain properties of the operators can be calculated explicitly. In this case $e^\pm=\left<v\right>, p^\pm=v_2$. Hence $\mu^+=\mu^-$, and we can therefore drop the $\pm$ in this example.
\begin{lem}
In the homogeneous case,
	$$
	\Proj\left[\gamma(v)h(x)\right]
	=
	\gamma(v)\cdot\frac{1}{P}\int_0^Ph(x)\;dx.
	$$
\end{lem}
\begin{proof}
This is straightforward, since $D$ reduces to the simple differential operator $\hat{v}_1\p_x$ in the homogeneous case, since there is equilibrium magnetic field.
\end{proof}

\begin{prop}
\begin{enumerate}
\item
Any homogeneous equilibrium that satisfies
	\begin{enumerate}
	\item
	$\int\mu_e\hat{v}_1^2\;dv<0$.
	\item
	$\int\mu_e\;dv\leq0$.
	\item
	$\int\mu_e\hat{v}_2^2\;dv-\int\mu_p\hat{v}_2\;dv<0$.
	\end{enumerate}
is unstable.

\item
There exists such an equilibrium.
\end{enumerate}
\end{prop}

\begin{proof}
\begin{enumerate}
\item
We verify that the above conditions verify the conditions of the first part of Theorem \ref{mainthm} which would imply that the solution is unstable.
Our plan is to show that the conditions listed in the proposition guarantee $l^0<0, neg\left(\Atwoz\right)\geq1$, and $neg\left(\Aonez\right)=0$.
Since $\mu$ has no $x$ dependence
	$$
	l^0
	=
	\frac{1}{P}\int_0^P\int\mu_e\hat{v}_1\Proj\left(\hat{v}_1\right)\;dv\;dx
	=
	\int\mu_e\hat{v}_1\Proj\left(\hat{v}_1\right)\;dv
	=
	\int\mu_e\hat{v}_1^2\;dv<0.
	$$

Next, we have
	$
	\Aonez h
	=
	-\p_x^2h-\left(\int\mu_e\;dv\right)h+\frac{1}{P}\int\mu_e\int_0^Ph\;dx\;dv,
	$
so that
	\begin{eqnarray*}
	\left<\Aonez h,h\right>
	&=&
	\int_0^P(h')^2\;dx-\left(\int\mu_e\;dv\right)\left[\int_0^Ph^2\;dx-\frac{1}{P}\left(\int_0^Ph\;dx\right)^2\right]\geq0
	\end{eqnarray*}
by H\"{o}lder's Inequality since $\int\mu_e\;dv\leq0$. Thus $neg(\Aonez)=0$. As for $\Atwoz$, we write
	$$
	\Atwoz h
	=
	-\p_x^2h-\left(\int\hat{v}_2\mu_p\;dv\right)h+\frac{1}{P}\int\hat{v}_2^2\mu_e\;dv\int_0^Ph\;dx.
	$$
Therefore, choosing $h\equiv1$, we have
	$$
	\left<\Atwoz 1,1\right>
	=
	-P\int\hat{v}_2\mu_p\;dv+P\int\hat{v}_2^2\mu_e\;dv
	<0,
	$$
so that $neg(\Atwoz)\geq1$.

This verifies that the conditions imply that $neg(\Atwoz)\geq1>0=neg(\Aonez)$, and that $l^0<0$. Therefore, by Theorem \ref{mainthm} such an equilibrium is unstable.

%
%
%
%
%

\item
We construct an explicit example.
Let $\mu(e,p)=\alpha(e)$ with $\alpha(e)=\gamma(e)+\eta(e)$, where $\gamma(e)=e-1$ on the interval $[1,2)$ and $0$ otherwise, and $\eta(e)=\exp{[-(e-2)^2]}$ on $[2,\infty)$ and $0$ otherwise. Even though $\alpha(e)$ is not smooth, it can be approximated by a smooth function that will still verify the calculations below. We verify the conditions of the proposition one by one:

	\begin{enumerate}
	\item
	We need to verify that $\int\mu_e\hat{v}_1^2\;dv<0$:
		$$
		\int\mu_e\hat{v}_1^2\;dv
		=
		\int\alpha'\;\hat{v}_1^2\;dv\\
		=
		\int(\gamma'+\eta')\;\hat{v}_1^2\;dv.
		$$
	We make the change of variables $v_1v_2\to r\theta$, so that $r^2=v_1^2+v_2^2=e^2-1$, to get
		\begin{eqnarray*}
		\int\mu_e\hat{v}_1^2\;dv
		&=&
		\int_0^{2\pi}\int_0^\infty(\gamma'+\eta')\frac{(r\cos\theta)^2}{1+r^2}\;r\;dr\;d\theta\\
		&=&
		\int_0^{2\pi}\cos^2\theta\;d\theta\left[\underbrace{\int_0^{\sqrt{3}}\frac{r^3}{1+r^2}\;dr}_{I}+\underbrace{\int_{\sqrt{3}}^\infty\eta'\frac{r^3}{1+r^2}\;dr}_{II}\right]
		\end{eqnarray*}
	Now:
		$$
		I
		=
		\int_0^{\sqrt{3}}\frac{r^3}{1+r^2}\;dr
		\approx
		0.8
		$$
		
		$$
		II
		=
		\int_{\sqrt{3}}^\infty\eta'\frac{r^3}{1+r^2}\;dr
		=
		-\int_{\sqrt{3}}^\infty2\left(\sqrt{1+r^2}-2\right)\exp{\left[-\left(\sqrt{1+r^2}-2\right)^2\right]}\frac{r^3}{1+r^2}\;dr
		\approx
		-2.5
		$$
	which, indeed, verifies the first condition.
	
	\item
	We verify that $\int\mu_e\;dv\leq0$.
		\begin{eqnarray*}
		\int\mu_e\;dv
		&=&
		\int\alpha'\;dv\\
		&=&
		\int_0^{2\pi}d\theta\left[\int_0^{\sqrt{3}}r\;dr+\int_{\sqrt{3}}^\infty\eta'\;r\;dr\right]\\
		&=&
		2\pi\left[\frac{3}{2}-\int_{\sqrt{3}}^\infty2\left(\sqrt{1+r^2}-2\right)\exp{\left[-\left(\sqrt{1+r^2}-2\right)^2\right]}\;r\;dr\right]\\
		&\approx&
		2\pi\left[\frac{3}{2}-2.9\right]<0.
		\end{eqnarray*}
	
	\item
	We have $\int\mu_e\hat{v}_2^2\;dv-\int\mu_p\hat{v}_2\;dv=\int\mu_e\hat{v}_2^2\;dv$ since $\mu$ does not depend on $p$. The proof that this integral is negative is identical to the above proof that $\int\mu_e\hat{v}_1^2\;dv<0$.

	\end{enumerate}
	
Thus our assumptions are all verified. This implies instability.
%
%
\end{enumerate}
\end{proof}

\subsection{Weak Magnetic Field}
\begin{prop}
There exists an inhomogeneous, nonmonotone, purely magnetic equilibrium that is unstable.
\end{prop}

Our goal is to construct an explicit purely magnetic equilibrium, for which we can conclude, using our main result, that it is unstable. Again, we would like to verify the conditions of Theorem \ref{mainthm}:
	$$
	l^0<0
	\hspace{1cm}
	\text{and}
	\hspace{1cm}
	neg\left(\Atwoz\right)>neg\left(\Aonez\right).
	$$
	
Therefore, we construct an equilibrium $\mu^\pm(e,p)$ for which $l^0<0, neg\left(\Atwoz\right)\geq1$ and $neg\left(\Aonez\right)=0$.
The main idea of the construction is to consider an equilibrium that is \emph{almost} monotone. We separate $x,v$ space into the sets $S^\pm_g$ (good) and $S^\pm_b$ (bad) where $\mu^\pm_e$ are negative or positive respectively. Then we show that if $S^\pm_b$ are not too big (in measure) we are essentially in the previously-known \emph{monotone} situation, and we can easily investigate the properties of the operators $\Aonez$ and $\Atwoz$.
To choose a magnetic potential we consider the ODE \eqref{psi-ode} for the magnetic potential, which can be written as
	\begin{equation}\label{4.2}
	\p_x^2\psi^0
	=
	2\int\hat{v}_2\mu^-(\left<v\right>,v_2-\psi^0(x))\;dv
	\end{equation}
using a simple change of variables, similar to the ones demonstrated in the proof of Lemma \ref{properties2}\eqref{bz}, and recalling that $\mu^+(e,p)=\mu^-(e,-p)$. For simplicity, we write \eqref{4.2} as
	$$
	\p_x^2\psi^0
	=
	g(\psi^0).
	$$
	
In \cite{rvm1} it is shown that a sufficient condition for the existence of a purely magnetic equilibrium is the existence of magnetic potential that solves this ODE. Indeed, there exist periodic solutions.
Moreover, assuming that $\mu^-(e,p)$ is even in $p$ we readily see that $g(0)=0$, and assuming that $\int\hat{v}_2\mu^-_p(\left<v\right>,v_2)\;dv>0$ we see that $g'(0)<0$. This implies that the origin is a center and that there exists $\epsilon_0>0$ for which there is a family of periodic solutions $\psi_\epsilon^0(x)$ with $0<\epsilon<\epsilon_0$ and periods $T_{\psi_\epsilon^0}$ depending upon $\epsilon$, satisfying \eqref{4.2}, with
	\begin{enumerate}
	\item
	$\left|\psi_\epsilon^0\right|_{C^1}\to0$ as $\epsilon\to0$,
	\item
	$T_{\psi_\epsilon^0}\to P_{cr}$ as $\epsilon\to0$. Here $P_{cr}$ is the period associated with the homogeneous case $\psi^0\equiv0$.
	\end{enumerate}
By readjusting the starting point, we may assume that $\psi_\epsilon^0$ obtains its minimum at $0$ and at $T_{\psi_\epsilon^0}$, that it has a single maximum, obtained at $\frac{1}{2}T_{\psi_\epsilon^0}$, that it is strictly increasing on $(0,\frac{1}{2}T_{\psi_\epsilon^0})$ and finally that it is symmetric with respect to $\frac{1}{2}T_{\psi_\epsilon^0}$.
We now show that for $\epsilon$ that is small enough, and with an appropriate choice of particle distribution $\mu^-(e,p)$, this distribution is linearly unstable.
As discussed above, we are flexible in our choice of particle distribution $\mu^-(e,p)$, as long as the following conditions are satisfied:
	
	\begin{enumerate}
	\item
	$\mu^-(e,p)$ is even in $p$
	
	\item
	$\int\hat{v}_2\mu^-_p(\left<v\right>,v_2)\;dv>0$
	
	\item
	$\sup{\mu^-_e}$ and $\left|S_b\right|$ satisfy the relation
		\begin{equation}\label{stabcond}
		\sup{\mu^-_e}<\frac{\pi^2}{3P^2_{cr}\left|S_b\right|}.
		\end{equation}
	where $\left|S_b\right|$ is the measure of the set $S_b=\left\{\mu^-_e>0\right\}$.
	\end{enumerate}
The first and third conditions are clearly easily satisfied with the right choice of $\mu^-(e,p)$. As for the second condition: Since $p$ is essentially $v_2$ (they differ by $\psiez$ which is a perturbation of $0$), the second condition essentially states that $\int p\mu^-_p\;dv>0$. This is satisfied since $\mu_p$ is an odd function of $p$.

We now show that $\Aonez$ is a nonnegative operator: Recall that

	$$
	\Aonez h=-\p_x^2h-\left(\int\mu^-_e\;dv\right)h+\int\mu^-_e\Proj h\;dv
	$$
where $\Proj$ is the projection operator onto $\ker D$. However,

	$$
	D
	=
	\hat{v}_1\p_x-\hat{v}_2B^0\p_{v_1}+\hat{v}_1B^0\p_{v_2},
	$$
where $B^0_\epsilon=\p_x\psiez$. Thus $\Proj$ and $\Aonez$ both depend on $\epsilon$. We denote them by $\Proj_\epsilon$ and $\Aoneze$. Let $K_\epsilon=\frac{\pi^2}{T_{\psiez}^2}$ be the constant given by Poincar\'{e}'s Inequality, which depends on $T_{\psiez}$ and thus on $\epsilon$. This constant satisfies
	$
	\int_0^{T_{\psiez}}h^2\;dx\leq K_\epsilon^{-1}\int_0^{T_{\psiez}}\left(\p_xh\right)^2\;dx.
	$
Since $T_{\psiez}$ depends on $\epsilon$ continuously, so does $K_{\epsilon}$. We let
	$
	K_0=\frac{\pi^2}{P_{cr}^2}
	$
be the Poincar\'{e} constant associated with the homogeneous case $\epsilon=0$. We have the following:

\begin{clm}
If \eqref{stabcond} holds then there exists some $\epsilon'>0$ such that $\Aoneze$ is nonnegative for all $0\leq\epsilon<\epsilon'$. 
\end{clm}

\begin{proof}
We first show for $\epsilon=0$, and then conclude for small values of $\epsilon$ by a certain continuity argument. For brevity we drop the $\epsilon$. Letting $u\in\hptz$ not identically $0$, we have:
	\begin{align*}
	\left<\Aonez u,u\right>_{\lpt}
	&=
	\int_0^P\left(\p_xu\right)^2\;dx
	-
	\sum_\pm\iint_{S^\pm_g}\mu^\pm_e\;dv\;u^2\;dx
	+
	\sum_\pm\iint_{S^\pm_g}\mu^\pm_e\Proj^\pm(u)\;u\;dv\;dx\\
	&\hspace{2.65 cm}-
	\sum_\pm\iint_{S^\pm_b}\mu^\pm_e\;dv\;u^2\;dx
	+
	\sum_\pm\iint_{S^\pm_b}\mu^\pm_e\Proj^\pm(u)\;u\;dv\;dx\\
	&=
	I-II+III-IV+V.
	\end{align*}
We need to show that this is positive for all $u$. Since $\Proj^\pm$ are projection operators, we easily have $|III|\leq|II|$ and $|V|\leq|IV|$. Since $II$ has a `good' sign, we have that $-II+III>0$. Terms $IV$ and $V$ have the `wrong' sign.
Using Poincar\'{e}'s inequality we can estimate:
	\begin{align*}
	|V|\leq IV
	=
	\sum_\pm\iint_{S^\pm_b}\mu^\pm_e\;dv\;u^2\;dx
	\leq
	\sum_\pm\left(\max_{S^\pm_b}\mu^\pm_e\right)\left|S^\pm_b\right|\int_0^Pu^2\;dx
	\leq
	\frac{K_0}{3} K_0^{-1}\int_0^P\left(\p_xu\right)^2\;dx
	=
	\frac{1}{3}I.
	\end{align*}
Thus $\left<\Aonez u,u\right>\geq I-IV+V\geq\frac{1}{3}I\geq\frac{K_0}{3}\left\|u\right\|_{\lpt}^2>0$. Finally, since $K_\epsilon$ varies continuously with $\epsilon$, we conclude that there must exist an $\epsilon'$ as in the claim.
\end{proof}

Next, we turn to the operator $\Atwoz$. Again, one should denote $\Atwoze$ to make clear the dependence upon $\epsilon$. It is shown in section 4.3 of \cite{rvm1} that for $\epsilon$ sufficiently small $\Atwoze$ has a negative eigenvalue if $\Atwozz$ has one. The same proof still holds in our situation.

Our last duty is to show that
	$
	l^{0,\epsilon}
	=
	\frac{1}{\Tpsiez}\int_0^{\Tpsiez}\int\hat{v}_1\mu^-_e\Proj_{\epsilon}\left(\hat{v}_1\right)dv\;dx$
is negative. By Lemma \ref{proj-parity} $\Proj$ preserves parity with respect to $v_1$. Thus $\hat{v}_1\Proj_{\epsilon}\left(\hat{v}_1\right)>0$. But since we chose $\mu^-(e,p)$ to be such that the domain where $\mu^-_e>0$ is negligible compared to the domain where $\mu^-_e<0$ (and $\mu^-_e$ is bounded by a small positive number from above), we can clearly pick $\mu^-(e,p)$ such that $l^0<0$.

\section{The Operators}\label{the-operators}
For the sake of completeness of this paper, we prove the important properties of our operators in full detail. These proofs appear in very similar form in \cite{rvm1,rvm2}. One significant difference is that we cannot use $\mu_e$ as a weight as it may vanish, and we therefore use the weight $w$ introduced in the introduction. Another notable novel part is Lemma \ref{properties1}\eqref{operators-positive} concerning the positivity of $\Aonel$ for large values of $\la$, which has a rather lengthy proof. Most of these lemmas are consequences of the properties of $\Ql_\pm$ discussed in Lemma \ref{propql}.

\begin{lem}[Properties of $\Aonel, \Atwol$]\label{properties1}
Let $0\leq\la<\infty$.
	\begin{enumerate}
	\item\label{op-norm}
	$\Aonel$ is selfadjoint on $\lptz$. $\Atwol$ is selfadjoint on $\lpt$. Their domains are $\hptz$ and $\hpt$, respectively, and their spectra are discrete.
	
	\item
	For all $h(x)\in\hptz$, $\|\Aonel h-\Aonez h\|_{\lpt}\to0$ as $\la\to0$. The same is true for $\Atwol$ with $h(x)\in\hpt$.

	\item
	For $i=1,2$ and $\sigma>0$, it holds that $\|\Ail-\Ais\|=O(|\la-\sigma|)$ as $\la\to\sigma$, where $\|\cdot\|$ is the operator 	norm from $\hptz$ to $\lpt$ in the case $i=1$, and from $\hpt$ to $\lpt$ in the case $i=2$.

	\item
	For all $h(x)\in\hptz$, $\|\Aonel h+\p_x^2h\|_{\lpt}\to0$ as $\la\to\infty$.

	\item
	When thought of as acting on $\hpt$ (rather than $\hptz$), the null spaces of $\Aonel$ and $\Aonez$ both contain the constant functions.
	
	\item\label{operators-positive}
	There exists $\gamma>0$ such that there exists $\La>0$ such that for all $\la\geq\La$, $\Ail>\gamma>0$, $i=1,2$.
	\end{enumerate}
\end{lem}

\begin{proof}
\begin{enumerate}
\item
Note that as mentioned in Remark \ref{remark} we are only interested in the action of $\Aonez,\Aonel$ on $\lptz$ and not $\lpt$, but that does not matter for the purpose of this lemma. We first show that the perturbations of the Laplacian in $\Ail$ are bounded operators for $i=1,2$ and all $\la\geq0$. For the case $\la>0$, a typical such perturbation may be estimated as follows:
	\begin{align*}
	\left\|\sum_\pm\int\mu^\pm_e\Ql_\pm h\;dv\right\|_{\lpt}
	&\leq
	\sum_\pm\left(\int_0^P\left|\int\mu^\pm_e\Ql_\pm h\;dv\right|^2dx\right)^{1/2}\\
	&\leq
	\sum_\pm\left(\int_0^P\left\{\int\left|\mu^\pm_e\right|dv\right\}\left\{\int\left|\mu^\pm_e\right|\left|\Ql_\pm h\right|^2dv\right\}dx\right)^{1/2}\\
	&\leq
	\sum_\pm\left(\sup_x\int\left|\mu_e^\pm\right|dv\right)\left(\int_0^P\left\{\int\left|\mu^\pm_e\right|\left|\Ql_\pm h\right|^2dv\right\}dx\right)^{1/2}\\
	&\leq
	\sum_\pm\left(\sup_x\int\left|\mu_e^\pm\right|dv\right)\left\|\Ql_\pm h\right\|_w\\
	&\leq
	\sum_\pm C^\pm\left\|\Ql_\pm h\right\|_w
	\leq
	C\left\|h\right\|_{\lpt}
	\end{align*}
due to the decay assumption \eqref{weight} and to the estimate \eqref{ql-estimate}. Here, $C$ is a universal constant, depending only on the equilibrium. The case $\la=0$ is even simpler, since the operators $\Ql_\pm$ are replaced by the projection operators $\Proj^\pm$ which clearly have operator norm $=1$.

As for the symmetry: Clearly, $-\p_x^2$ is symmetric. As for the other terms, we use the properties of the projections $\Proj^\pm$ and the operators $\Ql_\pm$ to show symmetry. In the case $\la=0$, an example of one of the terms is
	\begin{align*}
	\sum_\pm\int_0^P\int\mu^\pm_e\hat{v}_2\Proj^\pm(\hat{v}_2h)\;dv\;k\;dx
	&=
	\sum_\pm\int_0^P\int\frac{\mu^\pm_e}{w}\hat{v}_2\Proj^\pm(\hat{v}_2h)w\;dv\;k\;dx\\
	&=
	\sum_\pm\int_0^P\int\frac{\mu^\pm_e}{w}\Proj^\pm(\hat{v}_2h)\Proj^\pm(\hat{v}_2k)w\;dv\;dx
	\end{align*}
which is symmetric. (Note that in this calculation we multiply and divide by $w$ since the projection operators are defined in $\lwt$). In the case $\la>0$, we use the `almost-symmetry' property of $\Ql_\pm$ proved in Lemma \ref{propql}\eqref{almost-sym}: $\left<\Ql_\pm m,n\right>_w=\left<m,\Ql_\pm \tilde{n}\right>_w$, which is actually a symmetry property if $m$ and $n$ are functions of $x$ alone. This is precisely the case here.

The discreteness of the spectrum of $\Ail,\;i=1,2,\la\geq0$, is due to the fact that $-\p_x^2$ defined in $\lpt$ has discrete spectrum, and all other terms are bounded symmetric operators in $\lpt$, and, therefore, are relatively compact perturbations of $-\p_x^2$. Thus, according to the Kato-Rellich theorem \cite[V, \S4.1]{kato}, $\Aonel$ and $\Atwol$, $\la\geq0$, are selfadjoint with domains $\hptz$ and $\hpt$, respectively. Since $-\p_x^2$ has pure point spectrum, so do the operators $\Ail$, in view of Weyl's theorem \cite[IV, Theorem 5.35]{kato}.

\item
	We calculate
		$$
		\|\Aonel h-\Aonez h\|_{\lpt}
		=
		\left\|\sum_\pm\int\mu^\pm_e\left(\Ql_\pm h-\Proj^\pm h\right)dv\right\|_{\lpt}\\
		\leq
		C\|\Ql_\pm h-\Proj^\pm h\|_w\\
		\to
		0,
		$$
as $\la\to0$, by Lemma \ref{propql}\eqref{laz}, and using the fact that $|\mu^\pm_e|$ are bounded by $w$. Here, $C=\sum_\pm\sup_x\left(\int|\mu^\pm_e|\;dv\right)$. The proof for $\Atwol$ follows in precisely the same way.

\item
	Let $h(x)\in\hpt$. Then one has
		$$
		\|\Aonel h-\Aones h\|_{\lpt}
		=
		\left\|\sum_\pm\int\mu^\pm_e\left(\Ql_\pm h-\Qs_\pm h\right)dv\right\|_{\lpt}\\
		\leq
		C\sum_\pm\|\Ql_\pm h-\Qs_\pm h\|_w\\
		\leq
		C|\ln{\la}-\ln{\sigma}|\;\|h\|_w.
		$$
	Again, the proof for $\Atwol$ follows in precisely the same way.

\item
		$$
		\|\Aonel h+\p_x^2h\|_{\lpt}
		=
		\left\|\sum_\pm\int\mu^\pm_e\left(\Ql_\pm h-h\right)\;dv\right\|_{\lpt}
		\leq
		C\sum_\pm\|\Ql_\pm h-h\|_w
		\to
		0
		$$
	as $\la\to\infty$, by Lemma \ref{propql}\eqref{lainfty}. Here $C=\sum_\pm\sup_x\left(\int\left|\mu^\pm_e\right|\;dv\right)$.

\item
	This is clearly true by verifying that $\Aonez1=\Aonel1=0$ since $\Ql1=1$.

\item
	The proof for $\Atwol$ is straightforward: By Lemma \ref{propql}\eqref{qlnorm}, $\Ql_\pm$ both have operator norm $1$. Thus, the only unbounded terms in $\Atwol$ are $-\p_x^2$ and $\la^2$, both positive, and the claim follows.
	
	The proof for $\Aonel$ is more delicate. Since the following calculations are quite lengthy, we drop the $\pm$, but the same proof still stands when both species are considered. We need to show that for $\la$ sufficiently large and for $h\in\hpt$,
	\be
	0
	\leq
	\left<\Aonel h,h\right>
	=
	\int_0^P(\p_xh)^2\;dx
	+
	\int_0^P\int_{\R^2}\mu_e\left(\Ql h-h\right)h\;dv\;dx.
	\en
Since the Laplacian on a periodic domain has a certain spectral gap $G=G(P)>0$, we need to show that
	\be\label{need-to-show}
	\left|\int_0^P\int_{\R^2}\mu_e\left(\Ql h-h\right)h\;dv\;dx\right|\leq G.
	\en
Letting $-\infty<\beta<0$ be some constant to be chosen later, we may write the left-hand-side of \eqref{need-to-show} as
	\begin{align}
	\left|\int_0^P\int_{\R^2}\mu_e\left(\Ql h-h\right)h\;dv\;dx\right|
	\leq&
	\left|\int_{-\infty}^\beta\la e^{\la s}\int_0^P\int_{\R^2}\mu_e\left(h(X(s,x,v)-h(x)\right)h(x)\;dv\;dx\;ds\right|\\
	&+
	\left|\int_{\beta}^0\la e^{\la s}\int_0^P\int_{\R^2}\mu_e\left(h(X(s,x,v)-h(x)\right)h(x)\;dv\;dx\;ds\right|\nonumber
	\end{align}
which we denote by $I$ and $II$ respectively. We first analyze the term $I$:
	\begin{align*}
	|I|
	&=
	\left|\int_{-\infty}^\beta\la e^{\la s}\int_0^P\int_{\R^2}\mu_e\left(h(X(s,x,v)-h(x)\right)h(x)\;dv\;dx\;ds\right|\\
	&\leq
	\left(\int_{-\infty}^\beta\la e^{\la s}\;ds\right)\sup_{-\infty<s<\beta}\left|\int_0^P\int_{\R^2}\mu_e(h(X(s,x,v))-h(x))h(x)\;dv\;dx\right|\\
	&=e^{-\la|\beta|}\sup_{-\infty<s<\beta}\left|\int_0^P\int_{\R^2}\mu_e\left(\int_x^{X(s,x,v)}\p_\xi h(\xi)\;d\xi\right)h(x)\;dv\;dx\right|.\\
	\end{align*}
Now, given $s<0$ define
	\be
	\bar{s}:=\max_{t<0}\left\{t\;|\;X(t,x,v)\in\left\{X(s,x,v)+mP,\;m\in\Z\right\}\right\}.
	\en
Then due to the periodicity in the $x$ variable, we can replace $s$ by $\bar{s}$ in the last integral, and we therefore have:
	\begin{align*}
	|I|
	&\leq
	e^{-\la|\beta|}\sup_{-\infty<s<\beta}\left|\int_0^P\int_{\R^2}\mu_e\left(\int_x^{X(\bar{s},x,v)}\p_\xi h(\xi)\;d\xi\right)h(x)\;dv\;dx\right|\\
	&\leq
	e^{-\la|\beta|}\|h\|_{L^2}\sup_{-\infty<s<\beta}\left(\int_0^P\left(\int_{\R^2}\left(\int_x^{X(\bar{s},x,v)}\left|\p_\xi h(\xi)\right|\;d\xi\right)\;\frac{dv}{(1+\left<v\right>)^\alpha}\right)^{2}\;dx\right)^{1/2}.
	\end{align*}
Since $|\dot{X}|=|\hat{V}_1|<1$ we further have that $|X(\bar{s},x,v)-x|\leq P$. We can therefore finally estimate
	\begin{align*}
	|I|
	&\leq
	e^{-\la|\beta|}\|h\|_{L^2}\sup_{-\infty<s<\beta}\left(\int_0^P\left(\int_{\R^2}P^{1/2}\left(\int_x^{X(\bar{s},x,v)}\left|\p_\xi h(\xi)\right|^2\;d\xi\right)^{1/2}\;\frac{dv}{(1+\left<v\right>)^\alpha}\right)^{2}\;dx\right)^{1/2}\\
	&\leq
	e^{-\la|\beta|}\|h\|_{L^2}\sup_{-\infty<s<\beta}\left(\int_0^P\left(\int_{\R^2}P^{1/2}\left(\int_0^{P}\left|\p_\xi h(\xi)\right|^2\;d\xi\right)^{1/2}\;\frac{dv}{(1+\left<v\right>)^\alpha}\right)^{2}\;dx\right)^{1/2}\\
	&=
	Pe^{-\la|\beta|}\|h\|_{L^2}\|\p_xh\|_{L^2}\|(1+\left<v\right>)^{-\alpha}\|_{L^1}.\\
	\end{align*}

Now we turn to estimating the term $II$:
	\begin{align*}
	|II|=&\left|\int_{\beta}^0\la e^{\la s}\int_0^P\int_{\R^2}\mu_e\left(h(X(s,x,v)-h(x)\right)h(x)\;dv\;dx\;ds\right|\\
	&\leq
	\left(\int_{-\infty}^0\la e^{\la s}\;ds\right)\sup_{\beta\leq s\leq0}\left|\int_0^P\int_{\R^2}\left|h(X(s,x,v)-h(x)\right||h(x)|\;\frac{dv}{(1+\left<v\right>)^\alpha}\;dx\right|\\
	&\leq
	\|h\|_{L^2}\sup_{\beta\leq s\leq0}\left(\int_0^P\left[\int_{\R^2}\left|h(X(s,x,v)-h(x)\right|\;\frac{dv}{(1+\left<v\right>)^\alpha}\right]^2\;dx\right)^{1/2}\\
	&\leq
	\|h\|_{L^2}\sup_{\beta\leq s\leq0}\left(\int_0^P\left[\int_{\R^2}\int_x^{X(s,x,v)} \left|\p_\xi h(\xi)\right|\;d\xi\;\frac{dv}{(1+\left<v\right>)^\alpha}\right]^2\;dx\right)^{1/2}.\\
	\end{align*}
We now again use the fact that $|\dot{X}|=|\hat{V}_1|<1$, deducing that $|X(s,x,v)-x|<|s|\leq|\beta|$. Therefore if $|\beta|\leq P$, we have
	\begin{align*}
	II
	&\leq
	\|h\|_{L^2}\sup_{\beta\leq s\leq0}\left(\int_0^P\left[\int_{\R^2}\int_x^{X(s,x,v)} \left|\p_\xi h(\xi)\right|\;d\xi\;\frac{dv}{(1+\left<v\right>)^\alpha}\right]^2\;dx\right)^{1/2}\\
	&\leq
	\|h\|_{L^2}\sup_{\beta\leq s\leq0}\left(\int_0^P\left[\int_{\R^2}|s|^{1/2}\left(\int_x^{X(s,x,v)} \left|\p_\xi h(\xi)\right|^2\;d\xi\right)^{1/2}\;\frac{dv}{(1+\left<v\right>)^\alpha}\right]^2\;dx\right)^{1/2}\\
	&\leq
	\|h\|_{L^2}\sup_{\beta\leq s\leq0}\left(\int_0^P\left[\int_{\R^2}|s|^{1/2}\left(\int_0^P \left|\p_\xi h(\xi)\right|^2\;d\xi\right)^{1/2}\;\frac{dv}{(1+\left<v\right>)^\alpha}\right]^2\;dx\right)^{1/2}\\
	&=
	P^{1/2}\|h\|_{L^2}\|\p_x h\|_{L^2}\|(1+\left<v\right>)^{-\alpha}\|_{L^1}|\beta|^{1/2}.
	\end{align*}

Combining our estimates for the two terms $I$ and $II$ we have:
	\begin{align*}
	\left|\int_0^P\int_{\R^2}\mu_e\left(\Ql h-h\right)h\;dv\;dx\right|
	&\leq
	\|(1+\left<v\right>)^{-\alpha}\|_{L^1}\|h\|_{L^2}\|\p_xh\|_{L^2}\left(	
	Pe^{-\la|\beta|}
	+
	P^{1/2}|\beta|^{1/2}\right)\\
	&\leq
	\|(1+\left<v\right>)^{-\alpha}\|_{L^1}KG\left(	
	Pe^{-\la|\beta|}
	+
	P^{1/2}|\beta|^{1/2}\right)
	\end{align*}
where $K$ is the best constant given by Poincar\'{e}'s inequality on $[0,P]$. Recalling \eqref{need-to-show}, we need to choose $\beta$ and $\Lambda$ such that
	\be	
	P^{1/2}e^{-\la|\beta|}
	+
	|\beta|^{1/2}
	\leq
	\frac{1}{\|(1+\left<v\right>)^{-\alpha}\|_{L^1}KP^{1/2}}.
	\en
This is easily satisfiable by letting $|\beta|=\Lambda^{-1/2}$ and taking $\Lambda$ sufficiently large that only depends on $P$ and $\alpha$. We conclude that for for any $0<\gamma<G$, there exists $\Lambda>0$ such that for all $\la>\Lambda$, $\Aonel>\gamma>0$.

\end{enumerate}
\end{proof}

\begin{lem}[Properties of $\Bl, \Cl, \Dl$]\label{properties2}
Let $0<\la<\infty$.

	\begin{enumerate}
	\item\label{blnorm}
	$\Bl$ maps $\lpt\to\lpt$ with operator bound independent of $\la$.
	
	\item\label{bz}
	For all $h(x)\in\lpt$, $\|\Bl h\|_{\lpt}\to0$ as $\la\to0$. The same is true for $\Cl, \Dl$.

	\item
	If $\sigma>0$, then $\|\Bl-\Bs\|=O(|\la-\sigma|)$ as $\la\to\sigma$, where $\|\cdot\|$ is the operator norm from $\lpt$ to $\lpt$. The same is true for $\Cl, \Dl$.

	\item\label{bl-infty}
	For all $h(x)\in\lpt$, $\|\Bl h\|_{\lpt}\to0$ as $\la\to\infty$. The same is true for $\Cl, \Dl$.
	
	\end{enumerate}
\end{lem}

\begin{proof}
\begin{enumerate}
\item
Let $h,k\in\lpt$.

	\begin{eqnarray*}
	\left|\left<\Bl h,k\right>_{\lpt}\right|
	&\leq&
	\sum_\pm\left|\int_0^P\left(\int\mu^\pm_p\;dv\right)h(x)k(x)\;dx\right|
	+\sum_\pm\left|\int_0^P\left(\int\mu^\pm_e\Ql_\pm(\hat{v}_2h)\;dv\right)k(x)\;dx\right|\\
	&=&
	I+II
	\end{eqnarray*}
where	
	$$
	I
	\leq
	\sum_\pm\int_0^P\left|\left(\int\mu^\pm_p\;dv\right)h(x)k(x)\right|\;dx
	\leq
	\sum_\pm\sup_x\left(\int\left|\mu^\pm_p\right|\;dv\right)\left\|h\right\|_{\lpt}\left\|k\right\|_{\lpt}
	\leq
	C\left\|h\right\|_{\lpt}\left\|k\right\|_{\lpt},
	$$
and
	\begin{eqnarray*}
	II
	&\leq&
	\sum_\pm\left\{\int_0^P\left(\int\mu^\pm_e\Ql_\pm(\hat{v}_2h)\;dv\right)^2dx\right\}^{1/2}\left\{\int_0^P\left|k(x)\right|^2\;dx\right\}^{1/2}\\
	&\leq&
	\sum_\pm\left\{\int_0^P\left[\int\left|\mu_e^\pm\right|\;dv\right]\left(\int\left|\mu^\pm_e\right|\left|\Ql_\pm(\hat{v}_2h)\right|^2\;dv\right)dx\right\}^{1/2}\left\|k\right\|_{\lpt}\\
	&\leq&
	\sum_\pm\left(\sup_{x}\int\left|\mu^\pm_e\right|\;dv\right)\left\|\Ql_\pm(\hat{v}_2h)\right\|_w\left\|k\right\|_{\lpt}
	\leq
	C\left\|\hat{v}_2h\right\|_w\left\|k\right\|_{\lpt}
	\leq
	C\left\|h\right\|_{\lpt}\left\|k\right\|_{\lpt}.\\
	\end{eqnarray*}
\item
To show that $\Bl\to0$ strongly in $\lpt$ we use the fact that $\Ql_\pm\to\Proj^\pm$ strongly as shown in Lemma \ref{propql}. We consider the two terms that make up $\Bl$ separately. As for the first term of $\Bl$, $\sum_\pm\int\mu^\pm_p(e,p^\pm)\;dv$: Since $\mu_{p}^+(e,p^+)=-\mu_{p}^-(e,-p^+)$ one has
	$$
	\int\left[\mu_p^+(e,p^+)+\mu_p^-(e,p^-)\right]dv
	=
	\int\left[-\mu_p^-(\left<v\right>,-v_2-\psiz)+\mu_p^-(\left<v\right>,v_2-\psiz)\right]dv.
	$$
The change of variables $v_2\to-v_2$ applied to $\mu_p^-(e,-p^+)$ yields cancellation.
Thus the first term of $\Bl$ is $0$. (We note that even though this term, which does not depend upon $\la$, vanishes, we keep it in this paper, as it arrises naturally from Maxwell's Equations.) The second term is slightly more involved. As mentioned above, since $\Ql_\pm\to\Proj^\pm$ strongly in $L^2_w$, we have that
	$$
	\sum_\pm\int\mu_e^\pm\Ql_\pm(\hat{v}_2h)dv\to\sum_\pm\int\mu_e^\pm\Proj^\pm(\hat{v}_2h)dv
	$$
strongly in $\lpt$ as $\la\to0$. Therefore, we now show that $\sum_\pm\int\mu_e^\pm\Proj^\pm(\hat{v}_2h)dv=0$. Observe that since $E_1^0\equiv0$, $g(x,v)\in\ker{D^-}$ if and only if $g(x,-v)\in\ker{D^+}$. Therefore $\Proj^-[g(x,v)]=g(x,v)$ if and only if $\Proj^+[g(x,-v)]=g(x,-v)$. We conclude that $\Proj^-[g(x,v)](x,-v)=g(x,-v)=\Proj^+[g(x,-v)](x,v)$. Hence the second term of $\Bl h$ converges to
	\begin{eqnarray*}
	\int\left(\mu_e^+\Proj^+[\hat{v}_2h](x,v)+\mu_e^-\Proj^-[\hat{v}_2h](x,v)\right)dv
	&=&
	\int\left(\mu_e^+\Proj^-[-\hat{v}_2h](x,-v)+\mu_e^-\Proj^-[\hat{v}_2h](x,v)\right)dv\\
	&=&
	\int\left(-\mu_e^+\Proj^-[\hat{v}_2h](x,-v)+\mu_e^-\Proj^-[\hat{v}_2h](x,v)\right)dv
	\end{eqnarray*}
where in the second equality we used the linearity of $\Proj^-$. Now, we know that $\mu_e^+(e,p^+)=\p_e\mu^+(\left<v\right>,v_2+\psi^0)=\p_e\mu^-(\left<v\right>,-v_2-\psi^0)$. Therefore, by changing variables $v\to-v$ in the first integrand we get exact cancellation once again. Therefore, indeed, $\Bl\to0$ strongly in $\lpt$.

The fact that $\|\Cl(b)\|_{\lpt}\to0$ (and similarly for $\Dl$) holds since (i) as $\la\to0$, $\|\Ql_\pm m-\Proj^\pm m\|_w\to0$, as we have seen in Lemma \ref{propql}, (ii) $\mu$ is \emph{even} in $v_1$, and, finally, (iii) the projection operators $\Proj^\pm$ preserve parity with respect to the variable $v_1$ by Lemma \ref{proj-parity}. Thus $\sum_\pm\int \mu^\pm_e\Proj^\pm(\hat{v_1})\;dv=0$. Combining these three facts, we get that $\|\Cl(b)\|_{\lpt}\to\|b\sum_\pm\int \mu^\pm_e\Proj^\pm(\hat{v_1})\;dv\|_{\lpt}=0$.

\item
The proof of this fact is identical to the proof for $\Aonel$.

\item

By the triangle inequality:
	$$
	\|\Bl h\|_{\lpt}
	\leq
	\left\|h\sum_\pm \int\left(\mu_p^\pm+\hat{v}_2\mu^\pm_e\right)dv\right\|_{\lpt}
	+\left\|\sum_\pm\int\mu^\pm_e\left(\Ql_\pm(\hat{v}_2h)-\hat{v}_2h\right)dv\right\|_{\lpt}
	=
	I+II
	$$
where $I$ vanishes due to \eqref{fact}. The term $II$ is controlled as follows:
	\begin{align*}
	&II
	\leq
	\sum_\pm\left(\int_0^P\left\{\int\left|\mu^\pm_e\left(\Ql_\pm(\hat{v}_2h)-\hat{v}_2h\right)\right|dv\right\}^2dx\right)^{1/2}\\
	&\leq
	\sum_\pm\left(\int_0^P\left\{\int\left|\mu^\pm_e\right|dv\int\left|\mu_e^\pm\right|\left|\Ql_\pm(\hat{v}_2h)-\hat{v}_2h\right|^2dv\right\}dx\right)^{1/2}\\
	&\leq
	\sum_\pm\left(\left[\sup_x\int\left|\mu^\pm_e\right|dv\right]\int_0^P\left\{\int\left|\mu_e^\pm\right|\left|\Ql_\pm(\hat{v}_2h)-\hat{v}_2h\right|^2dv\right\}dx\right)^{1/2}
	\leq
	C\sum_\pm\left\|\Ql_\pm(\hat{v}_2h)-\hat{v}_2h\right\|_w\to0
	\end{align*}
as $\la\to\infty$ by Lemma \ref{propql}.\\

As for $\Cl$, we summarize the argument:

	$$
	\left\|\Cl\right\|_{\lpt}
	=
	\left\|\sum_\pm\int\mu^\pm_e\Ql_\pm(\hat{v}_1)\;dv\right\|_{\lpt}
	\to
	\left\|\sum_\pm\int\mu^\pm_e\hat{v}_1\;dv\right\|_{\lpt}
	=
	0
	$$
as $\la\to\infty$ by Lemma \ref{propql} and the fact that $\mu$ is even in $v_1$. The same proof holds for $\Dl$.
\end{enumerate}
\end{proof}

\begin{lem}[Properties of $l^\la$]\label{propll}
Let $0<\la<\infty$.
	\begin{enumerate}
	\item
	$l^\la\to l^0$ as $\la\to0$.
	
	\item
	$l^\la$ is uniformly bounded in $\la$.
	
	\end{enumerate}
\end{lem}

\begin{proof}
	\begin{enumerate}
	\item
	This is an immediate consequence of the properties of $\Ql_\pm$. As we have demonstrated how to use these properties above, we do not repeat the proof.
	
	\item
	This is due to the fact that $\Ql_\pm$ have operator norm $=1$, the integrability properties of $\mu^\pm_e$, and the fact that $\left|\hat{v}_1\right|\leq1$.
	\end{enumerate}
\end{proof}

The following lemma lists the important properties of $\Ml$ -- all of which are inherited directly from the properties of the various operators it is made up of, as listed in Lemma \ref{properties1}. 
\begin{lem}[Properties of $\Ml$]\label{propml}
To simplify notation, we write $u^T$ for a generic element $\left(\phi,\psi,b\right)\in\hpt\times\hpt\times\R$.
\begin{enumerate}
\item
For all $\la\geq0$, $\Ml$ is selfadjoint on $\lpt\times\lpt\times\R$ with domain $\hpt\times\hpt\times\R$.

\item
For all $u^T\in\hpt\times\hpt\times\R$, $\|\Ml u-\Mz u\|_{\lpt\times\lpt\times\lpt}\to0$ as $\la\to0$.

\item
If $\sigma>0$, then $\|\Ml-\Ms\|\to0$ as $\la\to\sigma$, where $\|\cdot\|$ is the operator norm from $\hptz\times\hpt\times\R$ to $\lpt\times\lpt\times\lpt$.

\end{enumerate}
\end{lem}

\begin{proof}
\begin{enumerate}
\item
This is true due to the structure of $\Ml$ and the selfadjointness of $\Aonel$ and $\Atwol$ for $\la\geq0$.

\item
We note that

$$
\left(\Ml-\Mz\right)u
=
\matthree{-\Aonel+\Aonez}{\Bl}{\Cl}{\left(\Bl\right)^*}{\Atwol-\Atwoz}{-\Dl}{\left(\Cl\right)^*}{-\left(\Dl\right)^*}{-P\left(\la^2-l^\la\right)-Pl^0}
\left(\begin{array}{c}\phi\\\psi\\ b\end{array}\right).
$$

Thus

\begin{eqnarray*}
\|\left(\Ml-\Mz\right)u\|_{\lpt\times\lpt\times\lpt}
&\leq&
\|\left(-\Aonel+\Aonez\right)\phi\|+\|\Bl\psi\|+\|\Cl b\|\\
&&+\|\left(\Bl\right)^*\phi\|+\|\left(\Atwol-\Atwoz\right)\psi\|+\|\Dl b\|\\
&&+\|\left(\Cl\right)^*\phi\|+\|\left(\Dl\right)^*\psi\|+\|\left(-P\left(\la^2-l^\la\right)-Pl^0\right)b\|.
\end{eqnarray*}
where all norms on the right hand side are in $\lpt$.
Since each of the terms on the right hand side tends to $0$ as $\la\to0$, we are done.
\item
We want to show that $\|\Ml-\Ms\|\to0$ as $\la\to\sigma$. This, again, is true by virtue of the fact that this is true for each of the entries of $\Ml$ separately.
\end{enumerate}
\end{proof}

\textit{Acknowledgements:} This paper is part of the author's Ph.D. thesis, supervised by Professor Walter Strauss. It is a great pleasure to thank Professor Strauss for his patient and kind guidance, his generous advice and the numerous hours he spent going over all details during the writing of this work.


\begin{thebibliography}{99}
 
\bibitem{energy58} Bernstein, I. B.; Frieman, E. A.; Kruskal, M. D.;
Kulsrud, R. M., \textit{An energy principle for hydromagnetic stability
problems}. Proc. Roy. Soc. London. Ser. A. \textbf{244} (1958), 17--40.

\bibitem{bgk} Bernstein, I., Greene, J., Kruskal, M., \textit{Exact nonlinear plasma oscillations}. Phys. Rev. \textbf{108}, 3, 17--40, (1958).

\bibitem{courant-hilbert-1} Courant, Richard and Hilbert, David,\textit{\ Methods of Mathematical Physics}, Vol. 1, Interscience Publishers 1966.

\bibitem{davidson} Davidson, Ronald C.,\textit{\ Physics of Nonneutral
Plasmas}, Addison-Wesley 1990.

\bibitem{davidson-qin} Davidson, Ronald C. and Hong, Qin, \textit{Physics of
intense charged particle beams in high energy accelerators}, World
Scientific, 2001.

\bibitem{dendy} Dendy, R. O., \textit{Plasma Dynamics}, Oxford Science Publications, 1990.

\bibitem{diperna-lions-bol} DiPerna, Ronald, and Lions, Pierre-Louis, \textit{On the Cauchy Problem for Boltzmann Equations: Global Existence and Weak Stability}, Ann. Math. \textbf{130} (1989), 321--366.

\bibitem{diperna-lions-vm} DiPerna, Ronald, and Lions, Pierre-Louis, \textit{Global Weak Solutions of Vlasov-Maxwell Systems}, Comm. Pure Appl. Math., \textbf{42} (1989), 729--757.

\bibitem{friedberg-mhd} Friedberg, J. P., \textit{Ideal Magnetohydrodynamics}%
, Plenum Press, 1987.

\bibitem{gardner} Gardner, C.S. \textit{Bound on the energy available from a plasma}, Phys. Fluids, \textbf{6} (1963), 839--840.

\bibitem{glassey}
  Glassey, Robert,
  \textit{The Cauchy Problem in Kinetic Theory},
  SIAM,
  1996.

\bibitem{glassey-strauss} Glassey, Robert and Strauss, Walter, \textit{Singularity formation in a collisionless plasma could occur only at high velocities}, Arch. Ration. Mech. Anal. \textbf{92} (1986), no. 1, 59--90.

\bibitem{gsc-exist-neutral} Glassey, Robert and Schaeffer, Jack, \textit{Global existence for the relativistic Vlasov-Maxwell system with nearly neutral initial data}, Comm. Math. Phys. \textbf{199} (1988), no. 3, 353--384.

\bibitem{gsc1.5} Glassey, Robert and Schaeffer, Jack, \textit{On the "one
and one-half dimensional" relativistic Vlasov-Maxwell system,} Math. Methods
Appl. Sci. \textbf{13} (1990), no. 2, 169--179.

\bibitem{gsc2.5} Glassey, Robert and Schaeffer, Jack, \textit{The "two and
one-half dimensional" relativistic Vlasov-Maxwell system,} Comm. Math. Phys. 
\textbf{185} (1997), 257-284.

\bibitem{goedbload} Goedbloed, Hans and Poedts, Stefan, \textit{Principles
of Magnetohydrodynamics : With Applications to Laboratory and Astrophysical
Plasmas, }Cambridge University Press, 2004.

\bibitem{grad} Grad, Harold, \textit{The guiding center plasma}. 1967 Proc.
Sympos. Appl. Math., Vol. XVIII pp. 162--248, Amer. Math. Soc., Providence,
R.I.

\bibitem{guo1} Guo, Yan, \textit{Stable magnetic equilibria in collisionless
plasma, }\ Comm. Pure Appl. Math., Vol L, 0891-0933 (1997).

\bibitem{guo2} Guo, Yan, \textit{Stable magnetic equilibria in a symmetric
plasma, }Commun. Math. Phys., \textbf{200, }211-247 (1999).
%

\bibitem{gs2} Guo, Yan and Strauss, Walter, \textit{Instability of periodic
BGK equilibria, }Comm. Pure Appl. Math. Vol XLVIII, 861-894 (1995).

\bibitem{gs5} Guo, Yan and Strauss, Walter, \textit{Unstable
oscillatory-tail waves in collisionless plasmas,} SIAM J. Math. Anal., 
\textbf{30}, no. 5, 1076-1114 (1999).

\bibitem{gs3} Guo, Yan and Strauss, Walter, \textit{Unstable relativistic
BGK waves, }Comput. and Appl. Math., \textbf{18}, no. 1, 87-122 (1999)

\bibitem{gs4} Guo, Yan and Strauss, Walter, \textit{Unstable BGK solitary
waves and collisionless shocks, }Commun. Math. Phys. \textbf{195}, 249-265
(1998).

\bibitem{gs6} Guo, Yan and Strauss, Walter, \textit{Magnetically created
instability in a collisionless plasma, }J. Math. Pures. Appl., \textbf{79, }
no. 10, 975-1009 (2000).

\bibitem{marsden} Holm, D., Marsden, J., Ratiu, T., Weinstein, A., \textit{Nonlinear stability of fluid and plasma equilibria}, Physics Reports, \textbf{123}, no. 1-2, 1-116
(1985).


\bibitem{kato} Kato, Tosio, \textit{Perturbation Theory for linear operators}
, (2nd edition) Springer 1976.

\bibitem{kruskal} Kruskal, Martin, \textit{Hydromagnetics and the theory of
plasma in a strong magnetic field, and the energy principles for equilibrium
and for stability}. 1960 La th\'{e}orie des gaz neutres et ionis\'{e}s
(Grenoble, 1959) pp. 251--274 Hermann, Paris; Wiley, New York

\bibitem{kulsrud-energy} Kulsrud, Russell, \textit{General stability theory
in plasma physics}. 1964 Advanced Plasma Theory (Proc. Internat. School of
Physics "Enrico Fermi", Course XXV, Varenna, pp. 54--96 Academic Press, New
York

\bibitem{langmuir}Langmuir, Irving, \textit{Oscillations in Ionized Gases}, Proc. Natl. Acad. Sci.. Vol 14, \textbf{8} 627-637
(1928).

\bibitem{lmp} Lavel, G., Mercier, C. and Pellat, R.M., \textit{Necessity of
the energy principle for magnetostatic stability}, Nuclear Fusion \textbf{5}
, 156-158 (1965).

\bibitem{lin01} Lin, Zhiwu, \textit{Instability of periodic BGK waves, }
Math. Res. Letts., \textbf{8}, 521-534 (2001).

\bibitem{lin-cpam} Lin, Zhiwu, \textit{Nonlinear instability of periodic
waves for Vlasov-Poisson system}, Comm. Pure Appl. Math. Vol LVIII, 505-528
(2005).

\bibitem{rvm2}
  Lin, Zhiwu, and Strauss, Walter,
  \textit{A sharp stability criterion for the Vlasov-Maxwell system},
  Invent. Math. \textbf{173}, 497--546 (2008)
 
\bibitem{rvm1}
  Lin, Zhiwu, and Strauss, Walter,
  \textit{Linear stability and instability of relativistic Vlasov-Maxwell systems}
   Comm. Pure Appl. Math. \textbf{60}, 724--787 (2007)
 
\bibitem{lw-nonlinear} Lin, Zhiwu and Strauss, Walter, \textit{Nonlinear
stability and instability of relativistic Vlasov-Maxwell systems}, \textit{\ 
} Comm. Pure Appl. Math. \textbf{60}, 789--837 (2007)

\bibitem{newcomb} Newcomb, W. A., \textit{Lagrangian stability of MHD fluids}, Nucl. Fusion. Suppl. \textbf{2, }451.

\bibitem{nicholson} Nicholson. D. R.,\textit{\ Introduction to Plasma
Theory}, Wiley, 1983.

\bibitem{parks} Parks, George, \textit{Physics of space plasmas}, Second
Edition, Westview Press, 2004.

\bibitem{penrose} Penrose, Oliver, \textit{Electrostatic Instabilities of a Uniform Non-Maxwellian Plasma}
, Phys. Fluids, \textbf{3}, 258, 258--265 (1960).

\bibitem{pfaff} Pfaffelmoser, K., \textit{Global classical solutions of the Vlasov-Poisson system in three dimensions for general initial data}, J. Diff. Eqns. \textbf{95}, 281--303 (1992).

\bibitem{sturrock}
  Sturrock, P.,
  \textit{Plasma Physics - An introduction to the theory of astrophysical, geophysical, and laboratory plasmas}.
  Cambridge University Press,
  1994.
 
\bibitem{taylor 74} Taylor, J.B., \textit{Plasma containment and stability
theory}, Proceedings of the Royal Society of London. Series A, Mathematical
and Physical Sciences, Vol. \textbf{304}, No. 1478. (1968), pp. 335-360.

\bibitem{trievelpiece} Trivelpiece, A. W. and Krall, N. A., \textit{
Principles of Plasma Physics}, McGraw-Hill, 1973.

\bibitem{vlasov} Vlasov, Anatoly, \textit{On Vibration Properties of Electron Gas} (Russian),  J. Exp. Theor. Phys., Vol. \textbf{8}, No. 3. (1938), p. 291.

\end{thebibliography}
\end{document}